\documentclass[a4paper,11pt,reqno]{amsart}
 \normalsize

\usepackage[utf8]{inputenc}
\usepackage[T1]{fontenc}
\usepackage[numbers,sort&compress]{natbib}

\usepackage{lmodern}
\usepackage{textcomp}
\usepackage{stmaryrd}

\usepackage[british]{babel}

\usepackage[final,spacing]{microtype}

  \usepackage{geometry}
\usepackage{enumerate}
\usepackage[shortlabels]{enumitem}
\usepackage{afterpage}
\usepackage{lscape}
\usepackage{color}
\usepackage[dvipsnames]{xcolor}
\usepackage{setspace}
\usepackage{graphics}

\usepackage{float}
\usepackage{tikz}       
\usepackage{tikz-cd}
\usepackage[new]{old-arrows}

\usepackage{bbm}
\usepackage{amsmath}
\usepackage{amsfonts}
\usepackage{amssymb}
\usepackage{amsthm}
\usepackage{mathtools}
\usepackage{mathrsfs}
\usepackage{commath}  
\usepackage{siunitx}            
\usepackage{tensor}             

\usepackage{physics}          
\usepackage{simpler-wick}
\usepackage{slashed}
\usepackage{pythonhighlight}

\usepackage{url}
\newcommand{\doi}[1]{\href{https://doi.org/#1}{\texttt{DOI}:\nolinkurl{#1}}}
\newcommand\arxiv[1]{[\href{http://arxiv.org/abs/#1}{\path{#1}}]}

\usepackage[final]{hyperref}
\hypersetup{colorlinks= true, citecolor=Salmon, filecolor=green,linkcolor=RoyalBlue, urlcolor=WildStrawberry}
\usepackage[english]{cleveref}
\usepackage{orcidlink}

\usepackage[textsize=tiny]{todonotes}

\newtheorem{theorem}{Theorem}[section]
\newtheorem*{theorem*}{Theorem}
\newtheorem{lemma}[theorem]{Lemma}
\newtheorem{cor}[theorem]{Corollary}
\newtheorem{prop}[theorem]{Proposition}
\theoremstyle{definition}
\newtheorem{definition}[theorem]{Definition}
\theoremstyle{remark}
\newtheorem{remark}[theorem]{Remark}

\newcommand{\Da}{\mathcal{D}}

\newcommand{\Fa}{\mathcal{F}}

\newcommand{\Oa}{\mathcal{O}}

\newcommand{\Fs}{\mathscr{F}}

\newcommand*{\alg}[1]{\mathfrak{#1}}

\newcommand{\fg}{\ensuremath{\mathfrak{g}}}
\newcommand{\fG}{\ensuremath{\mathfrak{G}}}
\newcommand{\fh}{\ensuremath{\mathfrak{h}}}

\newcommand{\fn}{\ensuremath{\mathfrak{n}}}

\newcommand{\fp}{\ensuremath{\mathfrak{p}}}

\newcommand{\ft}{\ensuremath{\mathfrak{t}}}

\let\H\undefined
\DeclareMathOperator{\H}{H}

\DeclareMathOperator{\ad}{ad}
\DeclareMathOperator{\End}{End}

\DeclareMathOperator{\Sym}{Sym}

\DeclareMathOperator{\Der}{Der}

\DeclareMathOperator{\sgn}{sgn}

\DeclareMathOperator{\id}{\boldsymbol{1}}

\DeclareMathOperator{\MC}{MC}

\newcommand{\sfm}{\mathsf{m}}

\newcommand{\Z}{\ensuremath{\mathbbm{Z}}}

\newcommand{\C}{\ensuremath{\mathbbm{C}}}

\renewcommand{\SS}{\ensuremath{\textup{S}}}

\newcommand{\ot}{\otimes}
\newcommand{\bu}{{\bullet}}

\newcommand{\cf}{\textit{cf.}~}
\newcommand{\ie}{\textit{i.e.}~}

\newcommand{\eg}{\textit{e.g.}~}

\hypersetup {
    pdftitle = {(Twisted) canonical supermultiplets and their resolutions as open-closed homotopy algebras},
    pdfauthor = {Simon Jonsson},
    pdfsubject = {math-ph},
    pdfkeywords = {open-closed homotopy algebra, twisting, L∞-algebra, pure spinor superfields, supersymmetry, homotopy algebra}
}

\title[]{(Twisted) canonical supermultiplets and their resolutions as open-closed homotopy algebras}
\author{Simon Jonsson}
\address{Department of Physics, Astronomy and Mathematics\\ University of Hertfordshire \\Hatfield, AL10 9AB, United Kingdom}
\email{\texttt{\href{mailto:d.jonsson@herts.ac.uk}{d.jonsson@herts.ac.uk}},\orcidlink{0009-0001-7155-8496}}
\date{\today}

\begin{document}
 \thispagestyle{empty}
\begin{abstract}
  We argue that some supersymmetric multiplets can naturally be equipped with the structure of an open-closed homotopy algebra. This structure is readily described through the pure spinor superfield formalism, which in particular associates a ``canonical multiplet'' for each choice of supersymmetry algebra. We study the open-closed homotopy algebra associated to (twists of) (resolutions of) the canonical multiplet, and show that it fits into a span of open-closed homotopy algebras, extending results of Cederwall {\it et al.} \cite{cederwall2023canonical}. 
\end{abstract}
\maketitle

\setcounter{tocdepth}{2}
\tableofcontents
\newpage
\section{Introduction}
A well-known challenge in the study of supersymmetric theories is finding viable off-shell representations of the spacetime supersymmetry: in some descriptions, the symmetry algebra is only represented on shell, \ie after imposing the equations of motion. This difficulty motivated the advancement of the Batalin-Vilkovisky (BV) formalism \cite{Batalin:1977pb,Batalin:1981jr,Batalin:1984jr,Batalin:1984ss,Batalin:1985qj}, where the introduction of auxiliary fields facilitates an off-shell formulation. In the language of homotopy theory this means one moves to a different (in some suitable sense equivalent) model of the underlying theory.

Relatedly, finding/constructing supersymmetric theories (or even just supermultiplets) using standard techniques is usually a case-by-case endeavour reliant on insight. Therefore, developing systematic and uniform methods for constructing supersymmetric theories is desirable. 

The pure spinor superfield formalism (see for example \cite{Cederwallspinorial,Cederwall2013PureOverview,jonssonthesis, perspectivesonpurespin,equivalence} and references therein) tackles both these issues simultaneously, that is, it provides a systematic way to construct supersymmetric representations, where the action of supersymmetry is both strict and manifest. The formalism can naively be seen as a generalisation of the above procedure of introducing auxiliary fields. To be a bit more precise, given a supersymmetry algebra $\fp$, the formalism is a functor $A^\bu(-)$ from the category of equivariant modules over (a derived replacement of) the ring $\Oa_Y$ of regular functions on the Maurer--Cartan variety $Y$ of $\fp$ to (a certain subcategory of) the category of \emph{$\fp$-multiplets}, characterised by homotopy $\fp$-representations satisfying some ``locality'' axioms. In fact, it was shown in \cite{equivalence} that this functor defines an equivalence of dg-categories, so that \emph{any} supermultiplet can be obtained through the pure spinor superfield formalism. Moreover, for any $\fp$-multiplet obtained through the formalism, one can define an appropriate notion of a ``component-field multiplet''.

As emphasised in \cite{cederwall2023canonical}, for each supersymmetry algebra $\fp$, there exists a ``canonical multiplet''\footnote{This multiplet was called the ``tautological multiplet'' in \cite{saberi2021twisting}.} associated to it: the multiplet $A^\bu(\Oa_Y)$ associated to $\Oa_Y$ itself. 
Furthermore, the fact that $\Oa_Y$ is a ring endows the canonical multiplet with the structure of a commutative dg algebra, or in other words an $A_\infty$-algebra which is strict and commutative. By homotopy transfer, the commutative structure on the canonical multiplet gives rise to an $A_\infty$-algebra\footnote{We should mention here that the transferred structure will in particular be $C_\infty$, we will not make this distinction in the sequel.} structure on the component fields. In some examples, this structure corresponds to the (colour-stripped\footnote{In the sense of \cite{Jurco:2019yfd,Borsten:2021hua}.}) homotopy algebra encoding the interactions of the underlying BV theory \cite{perspectivesonpurespin}.

Now, the commutative structure on the canonical multiplet is not freely generated, simply because $\Oa_Y$ is not freely generated.\footnote{As $\Oa_Y$ is the ring of functions on the Maurer--Cartan variety of $\fp$ it can be seen as the quotient of a polynomial ring by the (quadratic) ideal defining the Maurer--Cartan equation.} It is interesting to ask whether it is possible to resolve the canonical multiplet by a freely generated commutative  dg algebra. This is indeed possible, and a procedure for how was described in \cite{cederwall2023canonical}.\footnote{See also \cite{Chesterman:2002ey,Berkovits2005TheSpinors,Aisaka2008PureSpectrum, Movshev2004OnTheories, movshev2006algebraic, Movshev:2009ba,MovshevBar, jonssonthesis, Cederwall2015SuperalgebrasFunctions} for related work.} There, the authors proceed by considering the \emph{Tate resolution} \cite{Tate} of $\Oa_Y$. The resolution is a freely generated commutative dg algebra quasi-isomorphic to $\Oa_Y$. Moreover, the resolution corresponds to a $\fp$-multiplet $\widetilde{A^\bu}$ which is also endowed with a (freely generated) commutative algebra structure. One of the main results of \cite{cederwall2023canonical} was the construction of a span of quasi-isomorphisms of $A_\infty$-algebras.

\begin{theorem*}[\cite{cederwall2023canonical}]
There is a span
    \begin{equation}
        \begin{tikzcd}[row sep = 1 ex]
          \label{eq:roofuntwistintro}
          & \widetilde{A^\bu} \ar[rd] \ar[ld] &  \\
                A^\bu(\Oa_Y) & & C^\bu(\fn)
        \end{tikzcd},
      \end{equation}
of quasi-isomorphisms of $A_\infty$-algebras. 
\end{theorem*}
(The definition of $C^\bu(\fn)$ will be recalled in \cref{sec:tate} below.) 
      \bigskip

      Given that we now have a class of $\fp$-multiplets which also carry a multiplicative structure, the canonical multiplet $A^\bu(\Oa_Y)$ being an important example, it is natural ask about compatibility between these two structures. This is the topic of the present paper. 
     
Consider an ordinary Lie algebra $\fg$, and a Lie algebra representation $\rho:\fg\to \End(M)$ of $\fg$ on a vector space $M$. If $M$ is equipped with the structure of an associative algebra, a suitable notion for $M$ to be an algebra over $\fg$ is to say that the morphism $\rho$ factors through
\begin{equation}
  \begin{tikzcd}
    \fg \arrow[r,"{\rho}"] \arrow[dr]& \End(M)\\
    & \Der(M)\arrow[u,hookrightarrow]
  \end{tikzcd},
\end{equation}
where $\Der(M)$ is the Lie algebra of derivations of the multiplication on $M$. An example is to consider the action of vector fields on the algebra of functions on a manifold.

 The appropriate homotopy generalisation of the notion of an associative algebra over a Lie algebra was introduced by Kajiura and Stasheff in \cite{Kajiura2006}, where they define the notion of an $A_\infty$-algebra over an $L_\infty$-algebra. They furthermore observe that this notion is equivalent to what they aptly term an open-closed homotopy algebra (motivated by the work of Zwiebach \cite{ZWIEBACH1998193} on open-closed string field theory). Open-closed homotopy algebras moreover provide a ``good'' homotopy theory, in that homotopy transfer is possible \cite{Kajiura2006,openclosedkoszul}.

 In this paper we observe that the canonical multiplet, with its commutative structure, is an example of an $A_\infty$-algebra over the (strict) $L_\infty$-algebra $\fp$. Homotopy transfer then implies that the component-field multiplet is also an $A_\infty$-algebra over $\fp$.

 Using the framework of $A_\infty$-algebras over $L_\infty$-algebras we extend the above theorem of \cite{cederwall2023canonical}. 
\begin{theorem*}[\Cref{thm:rooftopenclosed}]
There is a span
    \begin{equation}
        \begin{tikzcd}[row sep = 1 ex]
          \label{eq:roofopenclosedintro}
          & \widetilde{A^\bu} \ar[rd] \ar[ld] &  \\
                A^\bu(\Oa_Y) & & C^\bu(\fn)
        \end{tikzcd},
      \end{equation}
of quasi-isomorphisms of $A_\infty$-algebras over the (strict) $L_\infty$-algebra $\fp$. 
\end{theorem*}

Finally, since the pure spinor formalism provides a natural framework for twisting  (see for example \cite{saberi2021twisting,Hahner:2023kts,Jonsson:2024uyr}). It is thus natural to ask how the span \eqref{eq:roofopenclosedintro} behaves under twisting. As we shall see, the construction of \eqref{eq:roofopenclosedintro}, is easily extended to the twisted setting. Our main result is the following theorem.
\begin{theorem*}[\Cref{thm:maintwisted}]
There is a span
    \begin{equation}
        \begin{tikzcd}[row sep = 1 ex]
          \label{eq:rooftwistopenclosedintro}
          & \widetilde{A^\bu}_Q \ar[rd] \ar[ld] &  \\
                A^\bu(\Oa_Y)_Q & & C^\bu(\fn_Q)
        \end{tikzcd},
      \end{equation}
of quasi-isomorphisms of $A_\infty$-algebras over the (strict) $L_\infty$-algebra $\fp_Q$.
\end{theorem*}
 (Here, $Q$ represents a Maurer--Cartan element of the supersymmetry algebra $\fp$, and the subscript $_Q$ denotes an appropriate deformation of the structure with respect to $Q$. )

\bigskip
The paper is organised as follows. In \cref{sec:preliminaries} we fix conventions related to dg vector spaces, and recall some results about the homological perturbation lemma,  (twisting of) (representations of) $L_\infty$-algebras, and (twisted) open-closed homotopy algebras. We also give a brief account of homotopy transfer of these structures. In \cref{sec:three} we go on to shortly describe how supersymmetric representations (supermultiplets) are produced through the pure spinor superfield formalism, and furthermore argue that certain (twisted) supermultiplets naturally carry an open-closed homotopy algebra structure. We then focus on the case of the (twisted) canonical multiplet and explicitly construct the spans \eqref{eq:roofopenclosedintro} and \eqref{eq:rooftwistopenclosedintro} above.

\section{Preliminaries}
\label{sec:preliminaries}
\subsection{(dg) vector spaces and algebras}
\label{sec:conventions}
Throughout we work over $\C$, the field of complex numbers. Tensor products are taken over $\C$ unless otherwise specified. Differential graded (dg) vector spaces and their morphisms are $\Z$-graded. The dual of a graded vector space is taken to be the direct sum of the dual of each graded piece.

Throughout we work with the Koszul sign convention. We let $\tau_{V,W}:V\ot W\to W\ot V$ denote the \emph{flip map}. Similarly, for $\sigma\in \SS_n$, the symmetric group on $n$ elements, and a graded vector space $V$, we define the map $\sigma(-):V^{\ot n}\to V^{\ot n}$ on homogeneous elements as $v_1\ot\cdots\ot v_n\mapsto \epsilon(\sigma;v_1,\ldots,v_n)v_{\sigma(1)}\ot\cdots\ot v_{\sigma(n)}$, where $\epsilon(\sigma;v_1,\ldots,v_n)$ is the Koszul sign.

Given a graded vector space $V$, we denote the free graded symmetric algebra on $V$ by $\Sym^\bu(V)$, for which we let $\vee$ denote the graded-commutative product. Similarly, we denote the graded skew-symmetric algebra on a vector space $V$ by $\bigwedge ^\bu V$, and the graded skew-symmetric product by $\wedge$. Unless explicitly stated, a (skew-)symmetric map, is always taken to be graded (skew-)symmetric. 

A (commutative) dg algebra ((c)dga) $(A,d,m)$, is a dg vector space $(A,d)$, equipped with a (symmetric) bilinear map $m:A\ot A\to A$ such that $d$ is a derivation with respect to $m$. A \emph{semifree} cdga is a freely generated cdga, that is, a cdga whose underlying graded commutative algebra is the free graded symmetric algebra on some graded vector space $V$.

Let $[n]$ be the one-dimensional graded vector space concentrated in degree $-n$. For any vector space $V$ we define $V[n]\coloneqq[n]\otimes V$, so that $(V[n])^k\cong V^{n+k}$ as vector spaces.

We will sometimes consider graded vector spaces equipped with an additional grading called \emph{weight}. When we do, each graded piece $V^n$ of a graded vector space $V$ has a decomposition $V^n=\bigoplus_{p\in\Z}V^{n,p}$. We should remark that the weight grading do not affect potential Koszul signs, rather it just provides an additional bookkeeping tool.\footnote{This grading convention differs from, but is equivalent to, the conventions of most literature on the pure spinor superfield formalism. In our conventions, the degree and weight are, for those well versed in the pure spinor literature, the totalised (\ie ghost number) and (the negative of the) internal degree, respectively. Our conventions match the \emph{Tate bigrading} of \cite{cederwall2023canonical}.}  The aforementioned structures of dg vector spaces and (c)dgas are trivially generalised to this setting by assigning weight zero to the differential and multiplications, respectively. We adopt the convention that taking duals reverse the weight grading, thus;
\begin{equation}
  (V^*)^{n,p}=(V^{-n,-p})^*.
\end{equation}

\subsection{The homological perturbation lemma}
\label{sec:hpl}
A \emph{(deformation) retract} of dg vector spaces $(M,d_M)$ and $(N,d_N)$ is a diagram
\begin{equation}
  \label{eq:ogretract}
  \begin{tikzcd}
\arrow[loop left,"{h}"](M,d_M) \arrow[r, shift left=1ex, "{p}"] 
 &  \arrow[l, shift left=1ex, "{i}"]  (N,d_N)
\end{tikzcd}\,,
\end{equation}
where $i$ and $p$ are morphisms of dg vector spaces, and $h$ is a degree $-1$ map such that $ pi=\id_{N}$ and $\id_{M}-ip=d_{M}h+hd_M$.
If moreover the \emph{side conditions} $hh=hi=ph=0$ hold, the retract is called a \emph{strong} (deformation) retract. Similarly, for $(M,d_M)$ and $(N,d_N)$ associative dg algebras, a \emph{strong (deformation) retract of associative dg algebras} or a \emph{strong multiplicative retract} \cite{CHUANG2019130}, is a strong deformation retract of the underlying dg vector spaces, with the additional condition that $i$ and $p$ are algebra morphisms, and $h$ is a $(ip,\id)$-derivation.\footnote{If $(A,\cdot)$ is an associative algebra and $f,g\in \End(A)$, then $s\in \End(A)$ is called an $(f,g)$-derivation if $s(x\cdot y)=s(x)\cdot g(y)\pm f(x)\cdot s(y)$.} 
 The following result is then standard in the literature.
\begin{prop}[Homological perturbation lemma, \cite{CHUANG2019130, Berglund:0909.3485}]
    \label{hpl}
  Let \begin{equation}
  \label{eq:ogretract5}
  \begin{tikzcd}
\arrow[loop left,"{h}"](M,d_M) \arrow[r, shift left=1ex, "{p}"] 
 &  \arrow[l, shift left=1ex, "{i}"]  (N,d_N)
\end{tikzcd}\,,
\end{equation}
be a strong (multiplicative) retract. Let $x\in\End^1(M)$ be a perturbation to the differential $d_M$, that is $(d_M+x)^2=0$. If the sum $\sum_{k\geq 0}(xh)^k$ is well-defined, there exist maps
\begin{equation}
  \begin{split}
\label{eq:hpl}
    i'=\sum_{k\geq 0}(hx)^k \circ i,\quad p'=p\circ\sum_{k\geq 0}(xh)^k,\\
    d_N'=p\circ\sum_{k\geq 0}(xh)^k x\circ i,\quad h'=h\circ\sum_{k\geq 0}(xh)^k, 
  \end{split}
\end{equation}
such that
\begin{equation}
  \label{eq:ogretract6}
  \begin{tikzcd}
\arrow[loop left,"{h'}"](M,d_M+x) \arrow[r, shift left=1ex, "{p'}"] 
 &  \arrow[l, shift left=1ex, "{i'}"]  (N,d_N')
\end{tikzcd}\,
\end{equation}
is a strong (multiplicative) retract. 
\end{prop}
Given a dg vector space $(V,d)$, it is always possible \cite{Lodval} to define a deformation retract to the cohomology of $d$;
\begin{equation}
 \label{eq:ogretractminimal}
  \begin{tikzcd}
\arrow[loop left,"{h}"](V,d) \arrow[r, shift left=1ex, "{p}"] 
&  \arrow[l, shift left=1ex, "{i}"]  (\H^\bu(V,d),0)
\end{tikzcd}\, .
\end{equation}

\subsection{\texorpdfstring{\(L_\infty\)}{L∞}-algebras}
\label{rem:cochains}
$L_\infty$-algebras are a homotopy generalisation of the notion of a Lie algebra.
\begin{definition}
An \(L_\infty\)-algebra \((\mathfrak g,\mu_k)\) is a graded vector space \(\mathfrak g=\bigoplus_{i\in\Z}\mathfrak g^i\) together with a family of skew-symmetric, multilinear maps \(\mu_k\colon\mathfrak g^{\wedge k}\to\mathfrak g\) of degree \(2-k\) for \(k\in\{1,2,3,\dotsc\}\) that satisfy, for all $n\geq1$, the identity
  \begin{equation}
    \label{eq:homotopy_Jacobi_identities}
 0=\sum_{\substack{i+j=n\\\sigma\in\SS_n}}\frac{(-1)^j}{i!j!}(-1)^{\sgn(\sigma)}\mu_{j+1}\circ (\mu_i\ot \id^{\ot j})\circ \sigma(-),   
\end{equation}
where $\sigma(-)$ is the map defined in \cref{sec:conventions}
\end{definition}
We will oftentimes abbreviate algebras by just their underlying vector space, \eg we may refer to an $L_\infty$-algebra $(\fg,\mu_k)$ as just $\fg$, leaving the additional structure implicit. An $L_\infty$-algebra $\fg$ is called \emph{strict} if $\mu_k=0$ for all $k>2$, and \emph{minimal} if $\mu_1=0$. Note that, the relations \eqref{eq:homotopy_Jacobi_identities} imply that $(\fg,\mu_1)$ is a dg vector space. 

An $L_\infty$-algebra structure on a graded vector space $\fg$ can equivalently be described by the datum of a degree one square-zero derivation $d_\fg$ (\ie $d_\fg^2=0$) on the free algebra $C^\bu(\fg)\coloneqq\Sym^\bu(\fg^*[-1])$, preserving the augmentation (see for example \cite{Reinhold}). The semifree cdga $(C^\bu(\fg),d_\fg)$, is usually called the (Chevalley--Eilenberg) cochains of $\fg$. This correspondence between semifree cdgas and $L_\infty$-algebras is an instance of the Koszul duality between commutative algebras and Lie algebras.\footnote{The interested reader can look at \cite{Priddy,Floystad,BGS,Positselski1,Positselski2,operadginzburg,Lodval, BrianNatalie} for some selected works related to Koszul duality.} We say that $\fg$ and $C^\bu(\fg)$ are Koszul dual.
\begin{remark}
  \label{rem:curved}
An $L_\infty$-algebra $\fg$ whose cochains do not preserve the augmentation is equipped with a \emph{curvature}; a map $\mu_0:\C\to \fg$, of degree 2. Such algebras are called \emph{curved $L_\infty$-algebras}, their homotopy Jacobi identities are identical to \eqref{eq:homotopy_Jacobi_identities}, but where the sum is valid for all $n\geq0$. See for example \cite{Kraft} for further details on curved $L_\infty$-algebras.    
\end{remark}
A \emph{morphism of $L_\infty$-algebras} (or \emph{$L_\infty$-morphism})  $\Phi:\fg\rightsquigarrow \fh$ is a degree 0 morphisms of the corresponding cochains
\begin{equation}
  \Phi^*:C^\bu(\fh)\to C^\bu(\fg),
\end{equation}
such that $\Phi^*(1)=1$.
The definition of a morphism of $L_\infty$-algebras is equivalent to a family of skew-symmetric multilinear maps
\begin{equation}
  \phi^{(k)}:\fg^{\wedge k}\to \fh,
\end{equation}
of degree $1-k$ satisfying some quadratic coherence relations involving the brackets of $\fg$ and $\fh$. In particular, $\phi^{(1)}$ is a morphism of the underlying cochain complexes. The explicit relations are not of interests to us in this paper and we refrain from writing them out, the reader is referred to \cite{Kraft} for more details.

We call an $L_\infty$-morphism  an \emph{$L_\infty$-(quasi-)isomorphism} if the first component map is a (quasi-)isomorphism of the underlying cochain complexes.  
\subsubsection{Twisting \texorpdfstring{\(L_\infty\)}{L∞}-algebras}
\(L_\infty\)-algebras admit a notion of ``twist'' with respect to a Maurer--Cartan element; see for example \cite{twistingprocedure,Kraft} for reviews. In the definitions below, for an $L_\infty$-algebra $(\fg,\mu_k)$, we assume that $\fg$ is equipped with a complete descending filtration; that is, a filtration $\Fa\fg$ of  $\fg$, compatible with the $L_\infty$-structure, and such that
\begin{equation}
  \label{eq:filtation0}
 \fg= \Fa^0\fg\supset\Fa^{1}\fg\supset\Fa^{2}\fg \supset\cdots,
\end{equation}
and $\bigcap_k\Fa^k\fg=\{0\}$. 
\begin{definition}
Let \((\mathfrak g,\mu_k)\) be an \(L_\infty\)-algebra with a complete descending filtration. The space of Maurer--Cartan elements of $\fg$ is defined as
\begin{equation}
    \MC(\fg)\coloneqq \{Q\in \fg^1\,|\,\sum_{i=1}^\infty\frac1{i!}\mu_i(Q,\dotsc,Q)=0\}.
\end{equation}
\end{definition}
\begin{definition}[\cite{twistingprocedure}]
Let \((\mathfrak g,\mu_k)\) be an \(L_\infty\)-algebra with a complete descending filtration.
Let \(Q\in\mathfrak g^1\) be a Maurer--Cartan element.
The \emph{twist} of \(\mathfrak g\) with respect to \(Q\) is the \(L_\infty\)-algebra \((\fg_Q,\mu_k^Q)\), defined on the graded vector space \(\fg_Q=\fg\), with brackets \(\mu_k^Q:\fg_Q^{\wedge k}\to \fg_Q\) defined by
\begin{equation}\label{eq:twistedproducts}
   \mu_k^Q(x_1,\ldots,x_k)=\sum_{i\geq 0}\frac1{i!}\mu_{i+k}(Q,\ldots,Q,x_1,\dotsc,x_k).
\end{equation}
\end{definition}
\subsection{\texorpdfstring{\(L_\infty\)}{L∞}-representations}
\label{sec:mods}
The notion of a representation of (or module over) a Lie algebra generalizes to the setting of homotopy algebras as follows. Recall that a representation of a Lie algebra $\fg$ on a vector space $M$ is a Lie algebra homomorphism $\rho:\fg\to \End(M)$, where $\End(M)$ is given a Lie algebra structure through the commutator bracket. More generally, if $(M,d_M)$ is a dg vector space, then $\End(M)$ is given the structure of a dg Lie algebra, with differential $[d_M,-]$.
\begin{definition}
  Let $(\fg,\mu_k)$ be an $L_\infty$-algebra. An $L_\infty$-representation (or $L_\infty$-module, or $\fg$-module) of $\fg$ on a dg vector space $(M,d_M)$ is an $L_\infty$-morphism
  \begin{equation}
    \rho:\fg\rightsquigarrow (\End(M),[d_M,-]).
  \end{equation}
\end{definition}
The definition of an $L_\infty$-module is equivalent to a family of component maps (sometimes referred to as higher actions)
\begin{equation}
  \label{eq:higheractions}
  \rho^{(k)}:\fg^{\wedge k}\ot M\to M,
\end{equation}
of degree $1-k$, which, together with the $\mu_k$'s, satisfy some quadratic coherence relations (see \eg \cite{Kraft}, and references therein). We refer to an $L_\infty$-module as \emph{strict} if $\rho^{(k)}=0$ for $k>1$.

Similarly, to the definition of a morphism of $L_\infty$-algebras, a morphism $\phi:M\rightsquigarrow N$ of $L_\infty$-modules is a family of component maps $\phi^{(k)}:\fg^{\wedge k}\ot M\to N$, satisfying some quadratic coherence relations \cite{DOLGUSHEV}. Furthermore, if the zeroth component is a (quasi-)isomorphism of the underlying cochain complexes, $\phi$ is called a (quasi-)isomorphism of $L_\infty$-modules.

\begin{remark}
  \label{remark:equivdefinitions}
There are other equivalent definitions of modules of $L_\infty$-algebras. For example it can be shown that an $L_\infty$-algebra module can be encoded as a particular $L_\infty$-algebra structure on the direct sum $\fg\oplus M$, see \cite{lada} for details on this construction. The $L_\infty$-module coherence relations then correspond to the homotopy Jacobi identities \eqref{eq:homotopy_Jacobi_identities} of the $L_\infty$-algebra.
\end{remark}
\subsubsection{Twisting \texorpdfstring{\(L_\infty\)}{L∞}-representations}
\label{sec:twistmods}
Twisting of $L_\infty$-algebras generalises to $L_\infty$-representations \cite{Kraft}. Explicitly, let $\fg$ be  an \(L_\infty\)-algebra \(\fg\) and $M$ a $\fg$-representation with structure maps \(\rho^{(k)}\). The \emph{twist} of $M$ by an element $Q\in \MC(\fg)$, is the $\fg_Q$-representation $M_Q$, whose underlying graded vector space is that of $M$, but with structure maps of the form \cite{Kraft}
\begin{equation}\label{eq:twistedmodule}
    \rho_Q^{(k)}(x_1,\dotsc,x_k) \coloneqq \sum_{i\geq 0}\frac1{i!}\rho^{(i+k)}(Q,\dotsc,Q,x_1,\dotsc,x_k).
  \end{equation}

  \subsection{\texorpdfstring{\(A_\infty\)}{A∞}-algebras over \texorpdfstring{\(L_\infty\)}{L∞}-algebras}
  \label{sec:openclosed}
Given an algebra $M$ and a representation $\rho:\fg\to \End(M)$, of a Lie algebra $\fg$ on the underlying vector space of $M$,  a suitable notion for $M$ to be an algebra over $\fg$ is to say that the morphism $\rho$ factors through
\begin{equation}
  \begin{tikzcd}
    \label{eq:derivation1}
    \fg \arrow[r,"{\rho}"] \arrow[dr]& \End(M)\\
    & \Der(M)\arrow[u,hookrightarrow]
  \end{tikzcd},
\end{equation}
where $\Der(M)$ is the Lie algebra of derivations of the multiplication on $M$. An example of this is the adjoint representation of $\fg$ on itself; because of the Jacobi identity $\fg$ acts through derivations. Another perhaps less trivial example is to consider the action of vector fields on the algebra of functions on a manifold.

In \cite{Kajiura2006}, motivated by the open-closed string field theory of Zwiebach \cite{ZWIEBACH1998193}, Kajiura and Stasheff generalise this construction to the case when $\fg$ is an $L_\infty$-algebra and $M$ is an $A_\infty$-algebra. Recall that, similarly to the definition of an $L_\infty$-algebra, an $A_\infty$-algebra on a vector space $A$ is a family of degree $2-k$ multilinear maps $m_k:A^{\ot k}\to A$, satisfying some quadratic coherence relations \cite{Lodval}.  $M$ is then called an \emph{$A_\infty$-algebra over an $L_\infty$-algebra}, and the way $\fg$ acts is through what are called \emph{homotopy derivations}\footnote{This was later generalised to homotopy derivations of $L_\infty$-algebras in \cite{Derofhomotopy} and to algebras over an operad in \cite{doubek2015homotopyderivations}.}; that is, the higher action maps are derivations with respect to each product, only up to homotopy.
\begin{definition}[\cite{Kajiura2006}]
  Let $(A,m_k)$ be an $A_\infty$-algebra. A \emph{homotopy derivation} of $A$, is a collection of degree $2-q$  multilinear maps
  \begin{equation}
    \label{eq:shd}
    \theta^q:A^{\ot q}\to A, \quad q\geq1,
  \end{equation}
  satisfying coherence relations of the form
  \begin{equation}
    \begin{split}
    \label{eq:shdcoherence}
      0= \sum_{r+s=q+1}\sum_{i=0}^{r-1}\pm\theta^r \circ (\id^{\ot i}\ot m_s\ot \id^{\ot (q-(i+s))})\pm m_r \circ (\id^{\ot i}\ot \theta^s\ot \id^{\ot (q-(i+s))}).
      \end{split}
    \end{equation}
    The explicit  signs can be found in \cite{Kajiura2006}.
\end{definition}
In particular, if $A$ is a strict $A_\infty$-algebra, with differential $d$ and multiplication $m$ (that is, an associative dga), \eqref{eq:shdcoherence} implies \cite{Kajiura2006}, for $q=1$,
\begin{equation}
  \begin{split}
  \label{eq:strictder}
    &\theta^1(m(x,y))-m(\theta^1(x),y)-(-1)^xm(x,\theta^1(y))                  =-d (\theta^2(x,y)) +\theta^2(dx,y)+(-1)^x\theta^2(x,dy).
    \end{split}
  \end{equation}
   We thus see that $\theta^1$ of \eqref{eq:strictder} is a derivation of $m$ only up to homotopy, controlled by $\theta^2$. We will call a homotopy derivation \emph{strict} if $\theta^k=0$ for all $k>1$.
\begin{definition}[\cite{Kajiura2006}]
  An \emph{$A_\infty$-algebra over an $L_\infty$-algebra} is the data of a graded vector space $A$, an $L_\infty$-algebra $(\fg,\mu_k)$, and a family of partially skew-symmetric multilinear maps
  \begin{equation}
    \label{eq:higherhigheraction}
    n_{p,q}:\fg^{\wedge p}\ot A^{\ot q}\to A, \quad p\geq0,\,q\geq 1,
  \end{equation}
 of degree $2-(p+q)$, satisfying coherence relations of the form
  \begin{align}
    \label{eq:openclosedcoherence}
    0= &\sum_{p+r=n}\sum_{\sigma \in \SS_n}\frac{\pm}{p!r!}n_{1+r,m}(\mu_{p}\ot \id_{\fg}^{\ot r}\ot\id_{A}^{\ot m})(\sigma(-) \ot \id_{A}^{\ot m})\\ \nonumber
    &+\sum_{\substack{p+r=n\\ i+s+j=m }}\sum_{\sigma \in \SS_n}\frac{\pm}{p!r!}n_{p,i+1+j} (\id_{\fg}^{p}\ot \id_{A}^{\ot i}\ot n_{r,s}\ot \id_A^{\ot j})(\id_{\fg}^{\ot p}\ot \tau_{\fg^{\ot r},A^{\ot i}}\ot \id_A^{j+s})(\sigma(-) \ot \id_{A}^{\ot m}),
  \end{align}
  where $\tau_{X,Y}$ denotes the flip map. The explicit signs can be found in \cite{Kajiura2006}. 
\end{definition}
As mentioned in \cite{Kajiura2006}, the definition of an $A_\infty$-algebra over an $L_\infty$-algebra encodes several substructures. In particular, the relations \eqref{eq:openclosedcoherence} imply that $(\fg\oplus A, \mu_1+n_{0,1})$ is a cochain complex. Furthermore, the maps $n_{0,q}$ define an $A_\infty$-algebra structure on $A$, and the maps $n_{p,1}$ define an $L_\infty$-representation of the $L_\infty$-algebra $\fg$ on the graded vector space $A$. The additional maps then correspond to higher maps of homotopy derivations, controlling the failure of the $n_{p,1}$ to act through strict derivations.

It is clear that, given a (dg) Lie algebra $\fg$, the above notion of an algebra over $\fg$ (\ie \eqref{eq:derivation1}) is a special case of an $A_\infty$-algebra over an $L_\infty$-algebra. Indeed, a (dg) Lie algebra is in particular an $L_\infty$-algebra, and a (dg) associative algebra is a special case of an $A_\infty$-algebra. Moreover, having $\fg$ act through (strict) derivations is equivalent to an $A_\infty$-algebra over an $L_\infty$-algebra where only ($\mu_1$,) $\mu_2$, $n_{1,1}$, ($n_{0,1}$,) and $n_{1,1}$ are nonzero.

Analogous to the fact that an $L_\infty$-module structure on $M$ of an $L_\infty$-algebra $\fg$ can be encoded as an $L_\infty$-algebra structure on the direct sum $\fg\oplus M$ (see \cref{remark:equivdefinitions}), it can be shown \cite{Kajiura2006} that the structure of an $A_\infty$-algebra $A$ over an $L_\infty$-algebra $\fg$ is equivalent to an algebraic structure on $\fg\oplus A$. In particular, the summand $\fg$ encodes the $L_\infty$-algebra structure on $\fg$ and the maps $n_{p,q}$ of \eqref{eq:higherhigheraction} are then seen as products on $\fg\oplus A$, satisfying the coherence relations of \eqref{eq:openclosedcoherence}. The data of such an algebraic structure is what Kajiura and Stasheff call an \emph{open-closed homotopy algebra}.
\begin{definition}[\cite{Kajiura2006}]
  An \emph{open-closed homotopy algebra} $(\alg{G}=\alg{G}_c\oplus\alg{G}_o,\mu_k,n_{p,q})$ is a graded vector space $\alg{G}=\alg{G}_c\oplus\alg{G}_o$, together with an $L_\infty$-algebra structure $(\alg{G}_c,\mu_k)$ on $\alg{G}_c$, and a family of partially skew-symmetric multilinear maps
   \begin{equation}
    \label{eq:higherhigheraction2}
    n_{p,q}:\alg{G}_c^{\wedge p}\ot \alg{G}_o^{\ot q}\to \alg{G}_o, \quad p\geq0,\,q\geq 1,
  \end{equation}
  of degree $2-(p+q)$, satisfying the coherence relations of \eqref{eq:openclosedcoherence}. 
\end{definition}

Similarly to how morphisms of $L_\infty$-algebras are defined by a family of component maps, morphisms of open-closed homotopy algebras are also defined through a family of component maps. Explicitly, a morphism $\Fs: (\alg{G}_c\oplus \alg{G}_o,\mu_k,n_{p,q})\rightsquigarrow (\alg{G}_c'\oplus \alg{G}_o',\mu_k',n_{p,q}')$ of open-closed homotopy algebras is a family of maps
\begin{equation}
  \Fs^{(k)}:(\alg{G}_c)^{\wedge k}\to \alg{G}_c',\quad  \Fs^{(p,q)}:(\alg{G}_c)^{\wedge p}\ot (\alg{G}_o)^{\wedge q}\to \alg{G}_o',
\end{equation}
for $k,q\geq1$ and $p\geq0$, satisfying quadratic coherence relations involving the products on the two open-closed homtopy algebras \cite{Kajiura2006}. In particular, the coherence relations say that $\Fs^{(1)}+\Fs^{(0,1)}$ is a morphism of the underlying cochain complexes, and so, analogous to $L_\infty$-algebras, a (quasi-)isomorphism of open-closed homotopy algebras is a morphism $\Fs$ of open-closed homotopy algebras, such that $\Fs^{(1)}+\Fs^{(0,1)}$ is a (quasi-)isomorphism \cite{Kajiura2006}.

As open-closed homotopy algebras are equivalent to $A_\infty$-algebras over $L_\infty$-algebras, we take the definition of a ((quasi-)iso)morphism of $A_\infty$-algebras over  $L_\infty$-algebras to be a ((quasi-)iso)morphism of the corresponding open-closed homotopy algebras.

\subsubsection{Twisting of open-closed homotopy algebras}
\label{sec:twistopenclosed}
Kajiura and Stasheff \cite{Kajiura2006} also provide an appropriate definition of Maurer--Cartan elements for an open-closed homotopy algebra. 
\begin{definition}[\cite{Kajiura2006}]
  Let  $\alg{G}=\alg{G}_c\oplus\alg{G}_o$ be an open-closed homotopy algebra, with $L_\infty$-algebra maps $\mu_k:\alg{G}_c^{\wedge k}\to \alg{G}_c$, and maps $n_{p,q}:\alg{G}_c^{\wedge p}\ot \alg{G}_o^{\ot q}\to \alg{G}_o$. An element $Q=(Q_c,Q_o)\in \alg{G}_c^1\oplus \alg{G}_o^1 $ is a \emph{Maurer--Cartan element of the open-closed homotopy algebra $\alg{G}$} if $Q_c$ is a Maurer--Cartan element of the $L_\infty$-algebra $\alg{G}_c$, and the following equation holds 
\begin{equation}
  \label{eq:openclosedmc}
  \sum_{\substack{p\geq 0\\q\geq1}}\frac1{p!}n_{p,q}(Q_c,\ldots,Q_c;Q_o,\ldots,Q_o)=0.
\end{equation}
\end{definition}
Using the definition of Maurer--Cartan elements of open-closed homotopy algebras one can then define the notion of twisting of open-closed homotopy algebras in an analogous fashion to twisting of $L_\infty$-algebras and their representations \cite{Kajiura2006}. As an open-closed homotopy algebra $(\alg{G}_c\oplus\alg{G}_o,\mu_k,n_{p,q})$ in particular encodes an $L_\infty$-representation of $\alg{G}_c$ on $\alg{G}_{o}$ this also includes the notion of twisting of $L_\infty$-representations. Explicitly, the element $(Q,0)$, where $Q$ is a Maurer--Cartan element of the $L_\infty$-algebra $\fG_c$, defines a Maurer--Cartan element of the open-closed homotopy algebra $\alg{G}$. The twisted open-closed homotopy structure corresponds to an open-closed homotopy algebra on $\alg{G}$, whose $L_\infty$-algebra structure is the twisted one of \eqref{eq:twistedproducts} and whose higher maps $n^Q_{p,q}$ are, for $x_1,\ldots,x_p\in\fG_c$ and $\sfm_1,\ldots,\sfm_q\in \fG_o$,
\begin{equation}
  \label{eq:twistopenclosed}
     n_{p,q}^{Q}(x_1,\ldots,x_p;\sfm_1,\ldots\sfm_q)=\sum_{k\geq0}\frac1{k!}n_{p+k,q}(Q,\ldots,Q,x_1,\ldots,x_p;\sfm_1,\ldots\sfm_q).
\end{equation}
We see that, for $q=1$, \eqref{eq:twistopenclosed} is nothing other than \eqref{eq:twistedmodule}, the component maps for the $L_\infty$-action of the twisted $L_\infty$-algebra $(\alg{G}_c)_Q$ on $\alg{G}_o$. 
\subsection{Homotopy transfer}
\label{sec:homotopytransfer}
Throughout the rest of the paper we will make heavy use of the homotopy transfer theorem \cite{Lodval}. Roughly what this says is that given a ``sufficiently nice''\footnote{More properly an algebra over a Koszul operad \cite{Lodval}.} algebraic structure on a dg vector space $\alg{G}$, and a deformation retract of dg vector spaces
\begin{equation}
\begin{tikzcd}
    (\mathfrak{G},d_{\mathfrak{G}}) \rar[shift left, "p"] \ar[loop left, "h"] & (\mathfrak{H},d_{\mathfrak{H}}) \lar[shift left, "i"],
\end{tikzcd}
\label{sdr}
\end{equation}
one can transfer this structure along the retract \eqref{sdr} to obtain an (in a suitable sense) equivalent algebraic structure on $\mathfrak{H}$. The structure of relevance in this paper is that of (representations of) $L_\infty$-algebras, $A_\infty$-algebras, and open-closed homotopy algebras, all of which are examples of ``sufficently nice'' algebraic structures.

Let us sketch how this works for $L_\infty$-algebras, the generalisation to other algebraic structures is straightforward. Let $\mathfrak{G}$ in \eqref{sdr} be an $L_\infty$-algebra $(\mathfrak{G},\mu_k)$, such that $\mu_1=d_{\mathfrak{G}}$. By virtue of the homotopy transfer theorem \cite{Lodval}, there exists an \(L_\infty\)-algebra structure on $\mathfrak H$ and an \(L_\infty\)-quasi-isomorphism
\begin{equation}\label{map to min}
    I\colon\mathfrak{H}\rightsquigarrow\mathfrak{G},
\end{equation}
such that \(I^{(1)}=i\). Furthermore, the $L_\infty$-algebra structure on $\mathfrak{H}$ and the $L_\infty$-quasi-isomorphism $I$ are expressible using the homotopy data \((i,p,h)\) and the $L_\infty$-algebra structure of $\fG$ \cite{Lodval}, \eg using the tensor trick \cite{Berglund:0909.3485}, and can be interpreted as a sum over trees. For example the ternary bracket $\mu_3^{\mathfrak{H}}$, is (modulo relative signs) the sum
\begin{equation}
  \begin{tikzpicture}[scale=0.5,baseline={([yshift=-1ex]current bounding box.center)}]
    \draw [thick] (-2,0) -- (-0.6,-1.4);
    \draw [thick] (-0.35,-1.65) -- (0,-2);
    \draw [thick] (2,0) -- (0,-2);
    \draw [thick] (0,-2) -- (0,-3);
    \draw [thick] (-1,-1)--(0,0);
    \node at (-.5,-1.5) {\tiny{$h$}};
    \node at (0,-3.5) {\tiny\(p\)};
    \node at (0,0.5) {\tiny\(i\)};
    \node at (-2,0.5) {\tiny\(i\)};
    \node at (2,0.5) {\tiny\(i\)};
    
     \fill[gray!50] (-1,-1) circle (0.35cm);
    \draw [thick] (-1,-1) circle (0.35cm);
    \node at (-0.98,-1) {\tiny$\mu_2$};
    \fill[gray!50] (0,-2) circle (0.35cm);
    \draw [thick] (0,-2) circle (0.35cm);
    \node at (0.02,-2) {\tiny$\mu_2$};
  \end{tikzpicture}
  +
  \begin{tikzpicture}[scale=0.5,baseline={([yshift=-1ex]current bounding box.center)}]
    \draw [thick] (-2,0) -- (0,-2);
    \draw [thick] (2,0) -- (0.65,-1.35);
    \draw [thick] (0,-2) -- (0.4,-1.6);
    \draw [thick] (0,-2) -- (0,-3);
    \draw [thick] (1,-1)--(0,0);
    \node at (0.55,-1.45) {\tiny{$h$}};
     \fill[gray!50] (1,-1) circle (0.35cm);
    \draw [thick] (1,-1) circle (0.35cm);
    \node at (1.02,-1) {\tiny$\mu_2$};
    \fill[gray!50] (0,-2) circle (0.35cm);
    \draw [thick] (0,-2) circle (0.35cm);
    \node at (0.02,-2) {\tiny$\mu_2$};\node at (0,-3.5) {\tiny\(p\)};
    \node at (0,0.5) {\tiny\(i\)};
    \node at (-2,0.5) {\tiny\(i\)};
    \node at (2,0.5) {\tiny\(i\)};
  \end{tikzpicture}
  +
  \begin{tikzpicture}[scale=0.5,baseline={([yshift=-1ex]current bounding box.center)}]
    \draw [thick] (-2,0) -- (-0.6,-1.4);
    \draw [thick] (-0.35,-1.65) -- (0,-2);
    
    \draw [thick] (1,-1) -- (0,-2);
    \draw [thick] (0,-3) -- (0,-2);
    \draw [thick] (1,-1)--(0,0);
    \draw [thick] (-1,-1) -- (0.3,-0.57);
    \draw [thick] (2,0) -- (0.7,-.45);
    \node at (-.5,-1.5) {\tiny{$h$}};
  
    \node at (0,-3.5) {\tiny\(p\)};
    \node at (0,0.5) {\tiny\(i\)};
    \node at (-2,0.5) {\tiny\(i\)};
    \node at (2,0.5) {\tiny\(i\)};

     \fill[gray!50] (-1,-1) circle (0.35cm);
    \draw [thick] (-1,-1) circle (0.35cm);
    \node at (-0.98,-1) {\tiny$\mu_2$};
    \fill[gray!50] (0,-2) circle (0.35cm);
    \draw [thick] (0,-2) circle (0.35cm);
    \node at (0.02,-2) {\tiny$\mu_2$};
  \end{tikzpicture}
  +
  \begin{tikzpicture}[scale=0.5,baseline={([yshift=-1ex]current bounding box.center)}]
    \draw [thick] (-2,0) -- (0,-2);
    \draw [thick] (2,0) -- (0,-2);
    \draw [thick] (0,0) -- (0,-2);
    \draw [thick] (0,-2) -- (0,-3);
    \node at (0,-3.5) {\tiny\(p\)};
    \node at (0,0.5) {\tiny\(i\)};
    \node at (-2,0.5) {\tiny\(i\)};
    \node at (2,0.5) {\tiny\(i\)};
    \fill[gray!50] (0,-2) circle (0.35cm);
    \draw [thick] (0,-2) circle (0.35cm);
    \node at (0.02,-2) {\tiny$\mu_3$};
  \end{tikzpicture}.
\end{equation}
More generally, $\mu_k^{\mathfrak{H}}$ is computed by a sum  over all rooted trees with $k$ leaves, where one decorates the leaves with $i$, the $n+1$-ary vertices with $\mu_n$, the internal edges with $h$, and the root with $p$. 

Similarly, if $\mathfrak{G}$ is an $A_\infty$-algebra, we, instead of summing over all rooted trees, sum over all planar trees to obtain an $A_\infty$-algebra structure on $\mathfrak{H}$ \cite{Lodval}.

It was shown in \cite{Kajiura2006} that homotopy transfer of open-closed homotopy algebras is possible\footnote{See also, \cite{openclosedkoszul}}. 
The way homotopy transfer works for open-closed homotopy algebras is analogous to that of $A_\infty$-algebras and $L_\infty$-algebras. However, now we assume that there is a splitting $\mathfrak{G}=\mathfrak{G}_c\oplus \mathfrak{G}_o$, and we are to sum over trees built from corollas of the form
\begin{equation}
  \begin{tikzpicture}[scale=0.6,baseline={([yshift=-1ex]current bounding box.center)}]
    \draw [thick,dashed] (-2,0) -- (0,-2);
    \draw [thick,dashed] (-1,0) -- (0,-2);
    \draw [thick] (0,0) -- (0,-2);
    \draw [thick] (1,0) -- (0,-2);
    \draw [thick] (2,0) -- (0,-2);
    \draw [thick] (0,-2) -- (0,-3);
    \node at (-2,0.5) {\tiny\(\mathfrak{G}_c\)};
    \node at (-1,0.5) {\tiny\(\mathfrak{G}_c\)};
    \node at (0,0.5) {\tiny\(\mathfrak{G}_o\)};
    \node at (1,0.5) {\tiny\(\mathfrak{G}_o\)};
    \node at (2,0.5) {\tiny\(\mathfrak{G}_o\)};
    \node at (0,-3.5) {\tiny\(\mathfrak{G}_o\)};
    \fill[gray!50] (0,-2) circle (0.45cm);
    \draw [thick] (0,-2) circle (0.45cm);
    \node at (0.02,-2) {\tiny$n_{p,q}$};
  \end{tikzpicture}\quad \text{, or}
    \begin{tikzpicture}[scale=0.6,baseline={([yshift=-1ex]current bounding box.center)}]
    \draw [thick,dashed] (-2,0) -- (0,-2);
    \draw [thick,dashed] (-1,0) -- (0,-2);
    \draw [thick,dashed] (0,0) -- (0,-2);
    \draw [thick,dashed] (1,0) -- (0,-2);
    \draw [thick,dashed] (2,0) -- (0,-2);
    \draw [thick,dashed] (0,-2) -- (0,-3);
    \node at (-2,0.5) {\tiny\(\mathfrak{G}_c\)};
    \node at (-1,0.5) {\tiny\(\mathfrak{G}_c\)};
    \node at (0,0.5) {\tiny\(\mathfrak{G}_c\)};
    \node at (1,0.5) {\tiny\(\mathfrak{G}_c\)};
    \node at (2,0.5) {\tiny\(\mathfrak{G}_c\)};
    \node at (0,-3.5) {\tiny\(\mathfrak{G}_c\)};
    \fill[gray!50] (0,-2) circle (0.45cm);
    \draw [thick] (0,-2) circle (0.45cm);
    \node at (0.02,-2) {\tiny$\mu_k$};
  \end{tikzpicture},
\end{equation}
where again we attach $h$ to the internal edges, $i$ to the leaves and $p$ to the root.

In fact, as open-closed homtopy algebras in particular encode the structure of a $L_\infty$-module on the underlying vector space of $\fG_o$, it is clear that homotopy transfer of open-closed homotopy algebras specialises to homotopy transfer of $L_\infty$-modules. Similarly, as open-closed homotopy algebras also encode an $A_\infty$-algebra structure on $\fG_o$, homotopy transfer of open-closed homotopy algebras specialises to homotopy transfer of $A_\infty$-algebras. 

Whenever \(\mathfrak{H}=\H^\bu(\mathfrak{G},d_{\mathfrak{G}})\) in \eqref{sdr}, then the corresponding homotopy algebra on $\H^\bu(\mathfrak{G},d_{\mathfrak{G}})$ is called a \emph{minimal model} of \(\mathfrak{G}\); minimal models are unique up to isomorphisms of the corresponding homotopy algebras \cite{Lodval}.

\section{Open-closed homotopy algebras in the pure spinor superfield formalism}
\label{sec:three}
Let us now apply the toolbox of open-closed homotopy algebras to physics. We will work in the pure spinor superfield formalism \cite{Cederwall2013PureOverview,perspectivesonpurespin, equivalence}. Let us briefly review how this construction works.

\subsection{Supersymmetry algebras}
\label{sec:susyalgebras}
Recall that a supersymmetry algebra can be defined as follows \cite{perspectivesonpurespin, equivalence}. Let $\ft$ be a Lie algebra in degrees one and two\footnote{To keep with current conventions in the literature we shall use the notation $\ft_1\coloneqq \ft^1$, and $\ft_2\coloneqq \ft^2$}. Then $\ft$ can be seen as a central extension of the abelian Lie algebra $\ft_1$ by $\ft_2$:
\begin{equation}
  \label{SES1}
   0 \to \ft_2 \to \ft \to \ft_1 \to 0\;.
 \end{equation}

 Choosing bases 
 $d_\alpha$ and $e_\mu$ for $\ft_1$ and $\ft_2$, respectively, the Lie bracket can  be written as
\begin{equation}
  [d_\alpha,d_\beta]=2\Gamma_{\alpha\beta}^\mu e_\mu,
\end{equation}
where $2\Gamma_{\alpha\beta}^\mu$ are the coefficients of the map $\ft_1\wedge\ft_1\to \ft_2$\footnote{Recall that $\wedge$ is the graded skew-symmetric product and since $\ft_1$ is in degree one this map is in fact symmetric.} defining the extension \eqref{SES1}. We call such an algebra a \emph{supertranslation algebra}. They cover any example of supersymmetry algebra with any amount of supersymmetry \cite{perspectivesonpurespin}. We can further extend $\ft$ by including the Lie algebra $\Der(\ft)$ of (even) derivations of $\ft$ in degree zero:
\begin{equation}
  \label{eq:superpoincare}
  \fp= \Der(\ft)\ltimes \ft.
\end{equation}
We call  an algebra of the form \eqref{eq:superpoincare} a \emph{super-Poincar\'e algebra} \cite{equivalence}. For future convenience, also define a weight grading ${\sf w}$, whereby we let ${\sf w}(\Der(\ft))=0$, ${\sf w}(\ft_1)=-1$, and ${\sf w}(\ft_2)=-2$. We can then define a $\Z\times \Z$-bigrading\footnote{We should mention here that our grading conventions differs from, but is equivalent to, the conventions appearing in most literature on the pure spinor superfield formalism. In our conventions, the degree and weight are, for those well versed in the pure spinor literature, the totalised (\ie ghost number) and (the negative of the) internal degree, respectively. Our conventions match the \emph{Tate bigrading} of \cite{cederwall2023canonical}.} 
\begin{equation}
  \label{eq:bidegree}
  (|\cdot|,{\sf w}(\cdot)).
\end{equation}
 We should clarify that the weight grading {\sf w} does not affect any potential Koszul signs; the Koszul signs are solely determined by the degree $|\cdot|$ of the objects. 

The Chevalley--Eilenberg cochains of $\ft$ take the form
\begin{equation}
  C^\bu(\ft)=\Sym^\bu(\ft_1^*[-1]\oplus\ft_2^*[-1])=\Sym^\bu(\ft_2^*[-1])\ot R,
\end{equation}
where $R =\Sym^\bu(\ft_1^*[-1])$. The generators of $R$, denoted $\lambda^\alpha$, sit in bidegree $(0,1)$, whilst the generators of $\Sym^\bu(\ft_2^*[-1])$, denoted $v^\mu$, sit in bidegree $(-1,2)$. In these generators the differential on $C^\bu(\ft)$ looks like $d_{\ft}=\lambda^\alpha\Gamma_{\alpha\beta}^\mu \lambda^\beta \pdv{}{v^\mu}$. The cochains can then be written as
\begin{equation}
  C^\bu(\ft)\cong\big(\C[\lambda^\alpha,v^\mu],\lambda^\alpha\Gamma_{\alpha\beta}^\mu \lambda^\beta \pdv{}{v^\mu}\big).
\end{equation}
In particular, the degree zero cohomology of this complex is the ring 
\begin{equation}\label{eq:oyascohomology}
  \H^0(C^\bu(\ft))=\Oa_Y = \C[\lambda^\alpha]/(\lambda^\alpha\Gamma_{\alpha\beta}^\mu \lambda^\beta),
\end{equation}
of functions on the space of Maurer--Cartan elements $Y\coloneqq\MC(\ft)$ of $\ft$. (Since $\ft$ has trivial differential, Maurer--Cartan elements are just elements $Q \in \ft_1$ satisfying $[Q,Q] = 0$.) Notice that, by virtue of \eqref{eq:oyascohomology}, there exists a map $C^\bu(\ft)\to \H^\bu(C^\bu(\ft))\to \Oa_Y$ witnessing $\Oa_Y$ as a
$C^\bu(\ft)$-module.

In \cite{perspectivesonpurespin, equivalence} the authors defined the notion of a \emph{$\fp$-multiplet} as a local $\fp$-module\footnote{that is, a vector bundle on whose sections $\fp$ acts through differential operators} with some additional structure. We will in the following only be concerned with $\fp$-module structures, as described in \cref{sec:mods}, although we believe most of the construction goes through when considering $\fp$-multiplets. For simplicity we will also work algebraically, \ie with polynomial functions, rather than in the smooth setting.
\subsection{The pure spinor construction}
Given a super-Poincar\'e algebra $\fp=\Der(\ft)\ltimes \ft$ we will use the notation $Y$ for the space of Maurer-Cartan elements. The \emph{pure spinor functor} \cite{perspectivesonpurespin, equivalence} $A^\bu(-)$ associates, to any $\Der(\ft)$-equivariant $C^\bu(\ft)$-module\footnote{A $C^\bu(\ft)$-module $M$ is said to be $\Der(\ft)$-equivariant if $M$ is also a $\Der(\ft)$-module, and the $C^\bu(\ft)$-action map is equivariant with respect to this structure, where $C^\bu(\ft)$ is a $\Der(\ft)$-module by the extension of the coadjoint action of $\Der(\ft)$ on $\ft^*$.} $M$, the corresponding $\fp$-module 
\begin{equation}
  \label{eq:purespinor}
  A^\bu(M)=\big(\C[x^\mu,\theta^\alpha]\ot M,\Da= \lambda^\alpha\pdv{}{\theta^\alpha}-\lambda^\alpha\Gamma_{\alpha\beta}^\mu\theta^\beta\pdv{}{x^\mu}+v^\mu\pdv{}{x^\mu}+d_M\big),
\end{equation}
where, by letting $x^\mu$ and $\theta^\alpha$  denote coordinate functions of $\ft_2$ and $\ft_1$, respectively,  $\C[x^\mu,\theta^\alpha]$ is the algebra of functions on $\ft$.\footnote{Physically thought of as flat superspace.} Thus, $\fp$ acts on $\C[x^\mu,\theta^\alpha]$. Specifically, $\Der(\ft)$ and $\ft$ acts by rotations and translations, respectively. Moreover, $\lambda^\alpha$ and $v^\mu$ act through the module structure on $M$ defined by the input data, and trivially on $\C[x^\mu,\theta^\alpha]$.  Even though we are only concerned with $\fp$-module structures in this work, we still refer to modules obtained through the pure spinor construction as ($\fp$-)multiplets. To avoid clutter in equations, when no confusion may arise we will oftentimes omit writing out the explicit indices.

The action of $\fp$ on $A^\bu(M)$ is the tensor product of the action on $\C[x,\theta]$ with the action on $M$ taken to be the trivial extension of the action of $\Der(\ft)$ to $\fp$. Explicitly, the generators of $\ft$ act as
\begin{equation}
  \begin{split}
    \rho(d_\alpha)&=\pdv{}{\theta^\alpha}+\Gamma_{\alpha\beta}^\mu\theta^\beta\pdv{}{x^\mu},\\
    \rho(e_{\mu})&= \pdv{}{x^\mu}.
    \end{split}
  \end{equation}
  We extend the bigrading \eqref{eq:bidegree} by assigning $\theta^\alpha$ and $x^\mu$ to have bidegree $(-1,1),\;(-2,2)$, respectively. Then the differential $\Da$ has bidegree $(1,0)$. 
  \subsubsection{The canonical multiplet}
  For each choice of super-Poincar\'e algebra there exists a \emph{canonical multiplet} \cite{cederwall2023canonical}\footnote{This was referred to as the ``tautological multiplet'' in \cite{saberi2021twisting}.} defined by taking as the $\Der(\ft)$-equivariant $C^\bu(\ft)$-module, the ring of functions $\Oa_Y$ on $Y$ (\cf \cref{sec:susyalgebras}). As $\Oa_Y$ is concentrated in degree 0, $v^\mu$ acts trivially and the canonical multiplet takes the form
\begin{equation}
  \label{eq:canonical}
  A^\bu(\Oa_Y)=\big(\C[x,\theta,\lambda]/(\lambda\Gamma\lambda), \Da=\lambda\pdv{}{\theta}-\lambda\Gamma\theta\pdv{}{x}\big).
\end{equation}
  
\subsubsection{Minimal models in the pure spinor construction}
A method was described in \cite{perspectivesonpurespin, equivalence}, for how to extract a ``component field'' multiplet (physically thought of as the space of component fields of a BV-theory) from the pure spinor multiplet $A^\bu(M)$. Let us briefly sketch how this works. For simplicity, let us assume that $M$ is concentrated in degree 0\footnote{The procedure is in full generality defined for any $\Der(\ft)$-equivariant $C^\bu(\ft)$-module $M$. However, when $M$ is concentrated in more than one degree, one proceeds in two steps. We refer the reader to \cite{equivalence} for details.}, then $v^\mu$ acts trivially, and the differential on $A^\bu(M)$ is
\begin{equation}
  \Da=\lambda\pdv{}{\theta}-\lambda\Gamma\theta\pdv{}{x}.
\end{equation}
Note that $\Da$ decomposes as $\Da=\Da_0+\Da_1$, where $\Da_0\coloneqq\lambda\pdv{}{\theta}$ squares to zero by itself. Thus $\Da_1=-\lambda\Gamma\theta\pdv{}{x}$, can be seen as a perturbation to $\Da_0$. We first consider the complex
\begin{equation}
  (A^\bu(M),\Da_0)=(\C[x,\theta]\ot M,\Da_0)=\C[x]\ot (\C[\theta]\ot M,\Da_0).
\end{equation}
The component-field multiplet is then obtained by first taking the cohomology of $\Da_0$. One can then define a deformation retract to the cohomology, and by introducing the perturbation $\Da_1$ one obtains, by the homological perturbation lemma, a retract
 \begin{equation}
    \begin{tikzcd}
       \label{eq:componentretract}
\arrow[loop left]{r}\big( A^\bu(M),\,\Da_0+\Da_1\big) \arrow[r, shift left=1ex, "{}"] 
 &  \arrow[l, shift left=1ex, "{}"]   \big(\C[x]\ot \H^\bu (\C[\theta]\ot M,\Da_0),\Da'\big)
\end{tikzcd}\, ,
\end{equation}
where $\Da'$ is the differential induced by the homological perturbation lemma of the deformation $\Da_1$. As described in \cite{perspectivesonpurespin} we can now transfer the $\fp$-module structure along \eqref{eq:componentretract} to obtain a $\fp$-module structure on the component fields. This $\fp$-module
\begin{equation}
  \label{eq:minimalamodel}
  \mu A^\bu(M)\coloneqq \big(\C[x]\ot \H^\bu (\C[\theta]\ot M,\Da_0),\Da'\big),
\end{equation}
can be taken as the component-field multiplet of $A^\bu(M)$ \cite{perspectivesonpurespin,equivalence}.
\subsubsection{Open-closed homotopy algebras in the pure spinor superfield formalism}
\label{sec:openclosedcanonical}
The canonical multiplet $A^\bu(\Oa_Y)$, carries some additional structure; it is in fact a cdga (or equivalently; a strict $A_\infty$-algebra, whose binary product is commutative), the multiplication taken to be the tensor product of the algebra structure on $\C[x,\theta]$ and $\Oa_Y$. It was explained in \cite{perspectivesonpurespin, equivalence} that, since $A^\bu(\Oa_Y)$ is in particular an $A_\infty$-algebra, letting $M=\Oa_Y$ in \eqref{eq:componentretract}, we can transfer this structure along \eqref{eq:componentretract} to obtain an $A_\infty$-algebra structure on $\mu A^\bu(\Oa_Y)$.\footnote{We should mention here that, since $A^\bu(\Oa_Y)$ is comutative, the transferred structure will in particular be $C_\infty$, we will not make this distinction in the sequel.} This $A_\infty$-structure corresponds, in some examples, to the (colour-stripped \cite{Jurco:2019yfd,Borsten:2021hua}) interaction terms for the underlying field theory. We refer the reader to \cite{perspectivesonpurespin, equivalence} and references therein for details on this.

However, there is in fact a richer structure present: as $\fp$ acts through (strict) derivations with respect to the product on  $A^\bu(\Oa_Y)$ we have that the canonical multiplet $A^\bu(\Oa_Y)$ is in fact an $A_\infty$-algebra over the $L_\infty$-algebra $\fp$. Equivalently, there is an open-closed homotopy algebra structure on $\fp\oplus A^\bu(\Oa_Y)$. Explicitly, in the notation of \cref{sec:openclosed}, $\mu_2$ is the Lie algebra (that is, strict $L_\infty$) structure on $\fp$, $n_{1,1}$ encodes the action of $\fp$ on $A^\bu(\Oa_Y)$, $n_{0,1}=\Da=\lambda\pdv{}{\theta}-\lambda\Gamma\theta\pdv{}{x}$, $n_{0,2}$ is the product on $A^\bu(\Oa_Y)$, and all other $n_{p,q}=0$. Consequently, since $A^\bu(\Oa_Y)$ is  an $A_\infty$-algebra over the  $L_\infty$-algebra $\fp$, or equivalently,  $\fp\oplus A^\bu(\Oa_Y)$ is an open-closed homotopy algebra, by homotopy transfer there is an open-closed homotopy algebra structure on $\fp\oplus \mu A^\bu(\Oa_Y)$. In particular, the induced products $n_{p,1}'$ will then encode the homotopy action of $\fp$ on $\mu A^\bu(\Oa_Y)$, and the maps $n_{0,k}'$ correspond to the $A_\infty$-algebra structure on $\mu A^\bu(\Oa_Y)$. However, in general there are also maps
\begin{equation}
  \label{eq:higherthetas}
  \theta_{p,q}\coloneqq n'_{p,q}:\fp^{\wedge p}\ot (\mu A^\bu(\Oa_Y))^{\ot q}, \quad p\geq 1,q \geq2.
\end{equation}
These maps then encode the compatibility between the homotopy action of $\fp$ and the $A_\infty$-algebra structure on $\mu A^\bu(\Oa_Y)$. This provides a precise framework for including interactions in a supersymmetric theory as discussed in for example \cite[Remark 2.1]{equivalence}. 
  
\begin{remark} One might ask what the appropriate physical interpretation of these $\theta_{p,q}$'s is. An $L_\infty$-module explains the appearance of linearised on-shell symmetries in physics (see for example \cite{perspectivesonpurespin}).  It is thus natural to expect that;
   just as the higher component maps of a homotopy action compensate for the fact that $\rho^{(1)}$ only closes modulo the linearised equations of motion, the $\theta_{p,q}$ should now, in the presence of interactions,  compensate for the fact that the symmetry algebra closes modulo the full equations of motion. 
\end{remark}

\begin{remark}
The above discussion is not limited to the canonical multiplet but works with any $\Der(\ft)$-equivariant $C^\bu(\ft)$-algebra, as long as the $\Der(\ft)$ action is through strict derivations with respect to the multiplication.
\end{remark}

\subsubsection{Twists of supermultiplets in the pure spinor formalism}
\label{sec:purespinortwist}
As the pure spinor construction produces strict $\fp$-modules, twisting in this formalism is simple: given a $\fp$-multiplet $A^\bu(M)$ from a $\Der(\ft)$-equivariant $C^\bu(\ft)$-module $M$, the twist of $A^\bu(M)$ by a Maurer--Cartan element $Q=Q^\alpha d_\alpha\in Y=\MC(\fp)$ is the $\fp_Q=(\fp,\ad_Q)$-module
\begin{equation}
  \label{eq:twistofcanonical}
  A^\bu(M)_Q=\big(\C[x,\theta]\ot M,\Da_Q= \Da+\rho(Q)\big),
\end{equation}
where $\rho(Q)=Q^\alpha\pdv{}{\theta^\alpha}+Q^\alpha\Gamma_{\alpha\beta}^\mu\theta^\beta\pdv{}{x^\mu}$, and $\Da$ is defined in \eqref{eq:purespinor}. We should mention here that the differential $\Da_Q$ no longer respects the weight grading ${\sf w}$, simply because $\rho(Q)$ does not. This issue is well-known, and can be resolved by introducing a formal variable $u$ of bidegree $(0,1)$, and letting $\Da_Q=\Da+u\rho(Q)$ see \eg \cite{saberi2021twisting,Costellotwistinglectures} for more details on this. We will, with slight abuse of notation, not write $u$ in equations and instead think of $Q^\alpha$ has having bidegree $(0,1)$. Since the action of $\fp$ on $A^\bu(M)$ is strict, the action of $\fp_Q$ on $A^\bu(M)_Q$ is still strict, and defined through the same map (\cf \cref{sec:twistmods}).

The twisted canonical multiplet $A^\bu(\Oa_Y)_Q$ is again a cdga; $\fp$ acts through derivations, and so the perturbed differential $\Da+\rho(Q)$ on $A^\bu(\Oa_Y)$ will still be a derivation. In \cite{saberi2021twisting}, these cdgas were studied in detail, for various examples of super-Poincar\'e algebras.

Moreover, since the action of $\fp_Q$ is defined through the same map as the action of $\fp$, and the underlying commutative algebra structure of $A^\bu(\Oa_Y)_Q$ is the same as that of $A^\bu(\Oa_Y)$, it is clear that $A^\bu(\Oa_Y)_Q$ is now an $A_\infty$-algebra over $\fp_Q$, or equivalently, there is an open-closed homotopy algebra structure on $\fp_Q\oplus A^\bu(\Oa_Y)_Q$. (This is equivalent to twisting the open-closed homotopy algebra structure on $\fp\oplus A^\bu(\Oa_Y)$ by the element $(Q,0)$, as described in \cref{sec:twistopenclosed}.)
\subsection{Resolving the canonical multiplet} 
As mentioned above, the canonical multiplet $A^\bu(\Oa_Y)$ is in particular a cdga. However, it is not freely generated as an algebra, simply because $\Oa_Y=\C[\lambda]/(\lambda\Gamma\lambda)$ is not. It is interesting however to try to further resolve $A^\bu(\Oa_Y)$ by a freely generated cdga. A procedure for how to go about this was described in \cite{cederwall2023canonical}.
There, the authors proceed by taking the \emph{Tate resolution} \cite{Tate} of $\Oa_Y$. Let us briefly recall how this works. 
\subsubsection{The Tate resolution}
\label{sec:tate}
Given a polynomial ring $R$ and a homogeneous ideal $I\subset R$, the Tate resolution \cite{Tate} constructs a multiplicative resolution of the quotient ring $R/I$ in free $R$-modules, that is, a semifree cdga over $R$. We construct the resolution iteratively, by ``killing cohomology'' degree by degree. The starting point for the resolution is $R$, we can endow it with the structure of a cdga by declaring that the differential is zero on all generators of $R$. (On free algebras, derivations are uniquely defined by their actions on generators.) The (degree zero) cohomology is then $R$. To construct a resolution of $R/I$ we now adjoin generators of degree $-1$ that are mapped by the differential to the generators of $I$, generating a semifree cdga whose degree zero cohomology is $R/I$. (Note that, for the case at hand here, the resulting cdga is just $C^\bu(\ft)$.)

However, with the introduction of new generators of degree $-1$, it may happen that $\H^{-1}$, the cohomology in degree $-1$, is non-zero. When this happens, we choose a set of generators of $\H^{-1}$, and adjoin degree $-2$ generators $w_2^i,\;i=1,\ldots,\dim(\H^{-1})$ to the cdga which are mapped to the generators of $\H^{-1}$ by the differential. The resulting semifree cdga then has empty cohomology in degree $-1$, while the degree zero cohomology is still $R/I$. Proceeding in this fashion by adding new generators to kill cohomology in negative degrees, the Tate resolution produces a semifree resolution of $R/I$. 

We call the procedure of introducing new generators a \emph{stage}. We number the stages so that at, stage $k$, generators in degree $-k$ are introduced. Thus, at stage $k$ we have a freely generated cdga whose degree zero cohomology is $R/I$, and whose cohomology in degrees $0 > d > -k$ is zero. ($R$ itself is then stage zero.) In the present example, the first stage is $C^\bu(\ft)$, the second stage adds generators to kill $\H^{-1}(C^\bu(\ft))$, and so on.

Recall that, in our example, $R=\C[\lambda^\alpha]$ and $I=(\lambda^\alpha\Gamma_{\alpha\beta}^\mu\lambda^\beta)$. Thus, both $R$ and $R/I=\Oa_Y$ are graded by weight. We can define the Tate resolution so as to respect this weight grading as follows. The generators adjoined at stage one (\ie the generators $v^\mu$ in $C^\bu(\ft)$) have weight two (as the generators $\lambda\Gamma^\mu\lambda$ of $I$ have weight two). Similarly, at higher stages, the weight of any new generator is equal to the weight of the cohomology it kills. The differential can then be seen as having bidegree $(1,0)$. 

The resulting algebra is a semifree cdga $\C[\lambda,v,w_2,w_3,\ldots]$, where $w_i$ are the generators adjoined at stage $i$, with a differential defined by the above procedure. By the discussion of \cref{rem:cochains} this algebra can then be interpreted as the Chevalley--Eilenberg cochains of a (minimal \ie zero differential) $L_\infty$-algebra $\tilde{\ft}$ (minimality follows from the construction of the Tate resolution, see \cite{cederwall2023canonical}). We thus refer to the resolution as $C^\bu({\tilde{\ft}})$. The $L_\infty$-algebra $\tilde{\ft}$ can be characterised \cite{cederwall2023canonical} by a short exact sequence of $L_\infty$-algebras
\begin{equation}
  \label{eq:tildeshortsequence}
  \begin{tikzcd}
    0\arrow{r}& \fn\arrow{r}&\tilde{\ft}\arrow{r} &\ft\arrow{r}&0,
  \end{tikzcd}
\end{equation}
where $\fn$ is the $L_\infty$-algebra defined through the Tate resolution from stage $2$ and onwards. The vector space $\fn$ is then spanned by the shifted duals of the $w_i$ for $i\geq2$. It defines an ideal in $\tilde{\ft}$ spanned by the generators of degree 3 and higher.

The weight grading defined through the Tate resolution defines a weight grading on $\tilde{\ft}$. Moreover, the weight grading is preserved by the brackets of the $L_\infty$-algebra $\tilde{\ft}$ (by construction the Chevalley--Eilenberg differential has weight zero). Taking the associated filtration to this gradation endows $\tilde{\ft}$ with a complete descending filtration.

Since the first stage of the Tate resolution is precisely $C^\bu(\ft)$, there is a natural map $C^\bu(\ft)\to C^\bu(\tilde{\ft})$  witnessing $C^\bu(\tilde{\ft})$ as a ($\Der(\ft)$-equivariant) $C^\bu(\ft)$-module. As such we can plug it in to $A^\bu$. The resulting multiplet;
\begin{equation}
  \label{eq:tildea}
  \widetilde{A^\bu}\coloneqq \big(A^\bu(C^\bu(\tilde{\ft}))=\C[x,\theta]\ot C^\bu(\tilde{\ft}),\,\tilde{d}= \lambda\pdv{}{\theta}-\lambda\Gamma\theta\pdv{}{x}+v\pdv{}{x}+d_{\tilde{\ft}}\big),
\end{equation}
(where $d_\ft$ is the Chevalley--Eilenberg differential on $C^\bu(\tilde{\ft})$) is a semifree cdga with a strict action of $\fp$.

Now, note that $A^\bu(\Oa_Y)$, $\widetilde{A^\bu}$, and $C^\bu(\fn)$ are all cdgas,  thus they are in particualr $A_\infty$-algebras. One of the main results of \cite{cederwall2023canonical} is the construction of a span of quasi-isomorphisms of $A_\infty$-algebras.
\bigskip
\begin{theorem}[\cite{cederwall2023canonical}]
  \label{thm:roofuntwisted}
There is a span
    \begin{equation}
        \begin{tikzcd}[row sep = 1 ex]
          \label{eq:roofuntwist}
          & \widetilde{A^\bu} \ar[rd] \ar[ld] &  \\
                A^\bu(\Oa_Y) & & C^\bu(\fn)
        \end{tikzcd},
      \end{equation}
of quasi-isomorphisms of $A_\infty$-algebras. 
\end{theorem}
As the pure spinor functor preserves quasi-isomorphisms \cite{equivalence}, it is clear that $A^\bu(\Oa_Y)$ and $\widetilde{A^\bu}$ are quasi-isomorphic as $\fp$-modules. Furthermore, by homotopy transfer along the right leg of \eqref{eq:roofuntwist}, there is a $\fp$-module structure on $C^\bu(\fn)$. This structure is, by the general theory of homotopy transfer \cite{Lodval}, quasi-isomorphic to the $\fp$-module structure on $\widetilde{A^\bu}$ and $A^\bu(\Oa_Y)$. 
Now, as we saw in \cref{sec:openclosedcanonical}, $A^\bu(\Oa_Y)$ is in fact an $A_\infty$-algebra over the (strict) $L_\infty$-algebra $\fp$. Analogously, so is $\widetilde{A^\bu}$. We will in the following show that \cref{thm:roofuntwisted} can be lifted to a span of $A_\infty$-algebras over $L_\infty$-algebras, and in fact show that this construction naturally extends to twisting of the multiplets.

\subsection{Spans of open-closed homotopy algebras}

\subsubsection{\texorpdfstring{$\widetilde{A}^\bu_Q$}{\~A\_Q} and the (twisted) canonical multiplet}
As $A^\bu(\Oa_Y)$ and $\widetilde{A^\bu}$ are both examples of $A_\infty$-algebras over $\fp$, there are open-closed homotopy algebra structures on $\fp\oplus A^\bu(\Oa_Y)$ and $\fp\oplus \widetilde{A^\bu}$, respectively. Moreover, because the action of $\fp$ is through (strict) derivations with respect to the multiplications on $A^\bu(\Oa_Y)$ and $\widetilde{A^\bu}$, the twists $A^\bu(\Oa_Y)_Q$ and $\widetilde{A^\bu}_Q$ by an element $Q^\alpha d_\alpha\in Y$ will be $A_\infty$-algebras over the $L_\infty$-algebra $\fp_Q$ (\cf \cref{sec:purespinortwist}). In the setting of open-closed homotopy algebras, after twisting there are open-closed homotopy algebra structures on $\fp_Q\oplus A^\bu(\Oa_Y)_Q$ and $\fp_Q\oplus \widetilde{A^\bu}_Q$, characterised by exactly the same maps as their untwisted counterparts, apart from the differentials, which are deformed according to the twisting procedure of $L_\infty$-algebras and their modules (see \cref{sec:preliminaries})\footnote{This is also equivalent to twisting the open-closed homotopy algebras by the element $(Q,0)$ as discussed in \cref{sec:purespinortwist}.}. Explicitly we have, as dg vector spaces,
\begin{align}
  \fp_Q\oplus A^\bu(\Oa_Y)_Q=&\big(\fp\oplus A^\bu(\Oa_Y),\;\ad_Q+\Da_Q\big),\; \text{}\; \Da_Q=\lambda\pdv{}{\theta}-\lambda\Gamma\theta\pdv{}{\theta}+\rho(Q),\\
  \fp_Q\oplus \widetilde{A^\bu}_Q=&\big(\fp\oplus \widetilde{A^\bu},\; \ad_Q+\tilde{d}_Q\big),\; \text{}\; \tilde{d}_Q=\tilde{d}+\rho(Q),
\end{align}
where $\rho(Q)$ is
defined below \eqref{eq:twistofcanonical}, and $\tilde{d}$ is defined in \eqref{eq:tildea}.
\begin{prop}
  \label{prop:AOYopenclosed}
  $\fp_Q\oplus \widetilde{A^\bu}_Q$ and  $\fp_Q\oplus A^\bu(\Oa_Y)_Q$ are quasi-isomorphic as open-closed homotopy algebras.
\end{prop}
\vspace{-4mm}
\begin{proof}
  We will proceed as in \cite{cederwall2023canonical}. Recall that, since $C^\bu(\tilde{\ft})$ is a resolution of $\Oa_Y$, we can find a strong retract
  \begin{equation}
    \begin{tikzcd}
       \label{eq:AOYretract}
\arrow[loop left]{r}{h} (C^\bu(\widetilde{\ft}),d_{\ft}) \arrow[r, shift left=1ex, "{p}"] 
 &  \arrow[l, shift left=1ex, "{i}"]  (\Oa_Y,0)  
\end{tikzcd}\, .
\end{equation}
Tensoring \eqref{eq:AOYretract} with $\C[x,\theta]$, and taking the direct sum with the trivial retract relating $\fp_Q$ with itself, yields a retract
  \begin{equation}
    \begin{tikzcd}
       \label{eq:AOYretract2}
\arrow[loop left]{r}{0+\id_{\C[x,\theta]} \ot h} (\fp_Q\oplus\C[x,\theta]\ot C^\bu(\widetilde{\ft}),\ad_Q+d_{\tilde{\ft}}) \arrow[r, shift left=1ex, "{\id_{\fp_Q}+\id_{\C[x,\theta]} \ot p}"] 
 &  \arrow[l, shift left=1ex, "{\id_{\fp_Q}+\id_{\C[x,\theta]}\ot i}"]  (\fp_Q\oplus \C[x,\theta]\ot\Oa_Y,\ad_Q+0)  
\end{tikzcd}\, .
\end{equation}
We now perturb the differential on the left of \eqref{eq:AOYretract2} by adding the twisted differential
\begin{equation}
  \widetilde{\Da}_Q\coloneqq v\pdv{}{x}+\lambda(\pdv{}{\theta}-\Gamma\theta\pdv{}{x})+Q(\pdv{}{\theta}+\Gamma\theta\pdv{}{x}).
\end{equation}
By the homological perturbation lemma, $\widetilde{\Da}_Q$ induces a differential on the summand $\C[x,\theta]\ot \Oa_Y\cong A^\bu(\Oa_Y)_Q$, for which the first term is
\begin{equation}
 (\id_{\C[x,\theta]}\ot p)\circ\widetilde{ \Da}_Q\circ(\id_{\C[x,\theta]}\ot i)=(\id_{\C[x,\theta]}\ot p)\circ\big(\lambda(\pdv{}{\theta}-\Gamma\theta\pdv{}{x})+Q(\pdv{}{\theta}+\Gamma\theta\pdv{}{x})\big)\circ (\id_{\C[x,\theta]}\ot i),
\end{equation}
which is precisely the form of the twisted differential $\Da_Q$ of the twist of the canonical multiplet \eqref{eq:twistofcanonical}. The higher deformations of the form $(\id_{\C[x,\theta]} \ot p)\circ(\widetilde{\Da}_Q(\id_{\C[x,\theta]}\ot h))^n\widetilde{\Da}_Q\circ(\id_{\C[x,\theta]}\ot i))$, are all zero by the same argument as in \cite[Prop. 3.2]{cederwall2023canonical}: the image of $h$ lies inside the ideal of $\C[x,\theta,\lambda,v,w_2,\ldots]$ generated by the elements of strictly negative degrees, that is, $\Im(h)\subseteq (v+\sum_{i\geq2}w_i)\C[x,\theta,\lambda,v,w_2,\ldots]\subset \C[x,\theta,\lambda,v,w_2,\ldots]$. Or in other words, the elements in the image of $h$ always contain at least one power of $v$ or $w_i$ for $i\geq 2$. By defining a weight grading ${\sf w}_2$ such that ${\sf w}_2 (x)={\sf w}_2 (\theta)={\sf w}_2 (\lambda)=0$ and ${\sf w}_2 (v)={\sf w}_2 (w_i)=1$, we have that the terms of $\widetilde{\Da}_Q$ either have ${\sf w}_2$ weight 0 or 1, and the ${\sf w}_2$ weight of each component of $h$ is always greater than or equal to 1. By slight abuse of notation we will write ${\sf w}_2(h)\geq1$ to mean that each component of $h$ has ${\sf w}_2$ weight at least one. Now, $p$ projects down to the ${\sf w}_2=0$ component and thus all terms with at least one $h$ must vanish. We thus end up with a deformation retract
\begin{equation}
    \begin{tikzcd}
      \label{eq:AOYtwistretract}
       \arrow[loop left]{r}{0+h'}  (\fp_Q\oplus\widetilde{A^\bu}_Q,\ad_Q+\tilde{d}_Q) \arrow[r, shift left=1ex, "{\id_{\fp_Q}+ p'}"] 
 &  \arrow[l, shift left=1ex, "{\id_{\fp_Q}+i'}"] (\fp_Q\oplus A^\bu(\Oa_Y)_Q,\ad_Q+\Da_Q)
\end{tikzcd}\, ,
\end{equation}
where
\begin{align}
  p'&=(\id_{\C[x,\theta]}\ot p)\circ \sum_{n\geq0}(\widetilde{\Da}_Q(\id_{\C[x,\theta]}\ot h))^n=(\id_{\C[x,\theta]}\ot p),\\
  i'&=\sum_{n\geq0}((\id_{\C[x,\theta]}\ot h)\widetilde{\Da}_Q)^n\circ(\id_{\C[x,\theta]}\ot i),\;\nonumber h'=(\id_{\C[x,\theta]}\ot h)\circ \sum_{n\geq0}(\widetilde{\Da}_Q(\id_{\C[x,\theta]}\ot h))^n.
\end{align}
Observe that all higher corrections of $\id_{\C[x,\theta]}\ot p$ are zero, by the same argument as for why the higher corrections to the differential are zero.

Now, the left side of \eqref{eq:AOYtwistretract} is equipped with the structure of an open-closed homotopy algebra. In particular, the $L_\infty$-algebra structure is given by that on $\fp_Q$, the $A_\infty$-algebra structure is given by that on $\widetilde{A^\bu}_Q$, and the action of $\fp_Q$, which is given by strict derivations, is just the strict one defined through the pure spinor construction. That is, all maps $n_{p,q}:\fp_Q^{\wedge p}\ot (\widetilde{A^\bu}_Q)^{\ot q}\to\widetilde{A^\bu}_Q$, apart from $n_{1,1}$, $n_{0,1}$,and  $n_{0,2}$, are zero. By homotopy transfer this gives rise to an open-closed homotopy algebra structure on $\fp_Q\oplus A^\bu(\Oa_Y)_Q$. However, there is already an open-closed homotopy algebra defined on this vector space; the one described in \cref{sec:openclosedcanonical}. To show that we indeed have an equivalence of $A_\infty$-algebras over an $L_\infty$-algebra we must show four things:
\begin{enumerate}
  \item The $L_\infty$-structure induced on $\fp_Q$, transferred along \eqref{eq:AOYtwistretract}, coincides with the one already defined.
\item The $A_\infty$-structure induced on $A^\bu(\Oa_Y)_Q$,  transferred along \eqref{eq:AOYtwistretract}, coincides with the one already defined.\\
\item The $\fp_Q$-module structure induced on $A^\bu(\Oa_Y)_Q$, transferred along \eqref{eq:AOYtwistretract}, coincides with the one already defined. \\
\item No higher maps $n'_{p,q}:\fp_Q^{\wedge p}\ot A^\bu(\Oa_Y)_Q^{\ot q}\to A^\bu(\Oa_Y)_Q$, for $p\geq 1$ and $q\geq 2$ appear through homotopy transfer along \eqref{eq:AOYtwistretract}.
\end{enumerate}
The first point is trivially checked. Indeed, the homotopy acting on $\fp_Q$ is just zero, and so, no higher brackets can appear. The binary bracket is just the one already defined.

The binary product induced on $A^\bu(\Oa_Y)$ is
\begin{equation}
  m'\coloneqq n'_{0,2}= (\id_{\C[x,\theta}\ot p)\circ m \circ (\id_{\C[x,\theta]}\ot i), 
\end{equation}
where $m\coloneqq n_{0,2}$ is the product on $\widetilde{A^\bu}$. This coincides with the binary product already defined on $A^\bu(\Oa_Y)$.

The linear part of the transferred action of $\fp_Q$ on $A^\bu(\Oa_Y)_Q$ coincides with that defined through the pure spinor construction: the action of $g\in\fp_Q$ on ${\sf x}\in A^\bu(\Oa_Y)_Q$ is given by
\begin{align}
  \label{eq:linearaction}\nonumber
  \rho'^{(1)}(g,{\sf x})=&(\id_{\C[x,\theta]}\ot p)\rho(g,(\id_{\C[x,\theta]}\ot i')({\sf x}))\\ \nonumber
  =& (\id_{\C[x,\theta]}\ot p)\rho(g,(\id_{\C[x,\theta]}\ot i)({\sf x}))\\
  &+(\id_{\C[x,\theta]}\ot p)\rho(g,\sum_{n\geq0}((\id_{\C[x,\theta]}\ot h)\widetilde{\Da}_Q)^n\circ (\id_{\C[x,\theta]}\ot i)({\sf x}))\\ \nonumber
  =&(\id_{\C[x,\theta]}\ot p)\rho(g,(\id_{\C[x,\theta]}\ot i)({\sf x})),
\end{align}
where we used that the second term of the second equality of \eqref{eq:linearaction} is zero:  ${\sf w}_2(h)\geq 1$ and  ${\sf w}_2(\rho)=0$, thus when projecting with $\id_{\C[x,\theta]}\ot p$ to ${\sf w}_2=0$ these terms are killed. The result of \eqref{eq:linearaction} coincides with the action of $\fp_Q$ on $A^\bu(\Oa_Y)_Q$ through the pure spinor construction.\footnote{It should perhaps  not come as a surprise that $\widetilde{A^\bu}_Q$ and $A^\bu(\Oa_Y)_Q$ are equivalent as $\fp_Q$-modules; the pure spinor functor preserves quasi-isomorphisms, and it can also be shown that twisting of $L_\infty$-modules preserves quasi-isomorophisms \cite{Kraft, DOLGUSHEV, Esposito}.}

Finally, to show that we actually have an equivalence of $A_\infty$-algebras over $L_\infty$-algebras, or open-closed homotopy algebras, the only thing left to show is that no higher maps of the form
  \begin{equation}
    \fp_Q^{\wedge p}\ot (A^\bu(\Oa_Y)_Q)^{\ot q}\to A^\bu(\Oa_Y)_Q,
  \end{equation}
  for $p=0$ and $q\geq 3$ (no higher $A_\infty$-products), $p\geq2$ and $q=1$ (no higher actions), and $p\geq1$ and $q\geq2$ (no higher $\theta$'s, \cf \eqref{eq:higherthetas}), can appear through homotopy transfer of open-closed homotopy algebras. Indeed all these maps vanish. To see this, we note that any such map can be represented by a diagram with a root of the form: 
  \begin{equation}
    \label{eq:trivialtrees}
    \begin{tikzpicture}[scale=0.59,baseline={([yshift=-1ex]current bounding box.center)}]
    \draw [thick, dashed] (-2,0) -- (0,-2);
    \draw [thick] (2,0) -- (0,-2);
    \draw [thick] (0,-2) -- (0,-3);
    \draw [thick] (1,-1)--(0,0);
     \fill[white!50] (0.5,-1.5) circle (0.26cm);
    \node at (0.59,-1.5) {\tiny{$h'$}};
     \fill[gray!50] (1,-1) circle (0.26cm);
    \draw [thick] (1,-1) circle (0.26cm);
    \node at (1.02,-1) {\tiny$m$};
    \fill[gray!50] (0,-2) circle (0.26cm);
    \draw [thick] (0,-2) circle (0.26cm);
    \node at (0.02,-2) {\tiny$\rho$};
    \node at (0,-3.5) {\tiny\(\id_{\C[x,\theta]} \ot p\)};
    \node at (0,0.5) {\footnotesize\(\widetilde{A^\bu}_Q\)};
    \node at (-2,0.43) {\footnotesize\(\fp_Q\)};
    \node at (2,0.5) {\footnotesize\(\widetilde{A^\bu}_Q\)};
  \end{tikzpicture}
  \text{, or}
      \begin{tikzpicture}[scale=0.59,baseline={([yshift=-1ex]current bounding box.center)}]
    \draw [thick] (-2,0) -- (0,-2);
    \draw [thick] (2,0) -- (0,-2);
    \draw [thick] (0,-2) -- (0,-3);
    \draw [thick] (1,-1)--(0,0);
     \fill[white!50] (0.5,-1.5) circle (0.26cm);
    \node at (0.59,-1.5) {\tiny{$h'$}};
     \fill[gray!50] (1,-1) circle (0.26cm);
    \draw [thick] (1,-1) circle (0.26cm);
    \node at (1.02,-1) {\tiny$m$};
    \fill[gray!50] (0,-2) circle (0.26cm);
    \draw [thick] (0,-2) circle (0.26cm);
    \node at (0.02,-2) {\tiny$m$};
    \node at (0,-3.5) {\tiny\(\id_{\C[x,\theta]}\ot p\)};
    \node at (0,0.5) {\footnotesize\(\widetilde{A^\bu}_Q\)};
    \node at (-2,0.5) {\footnotesize\(\widetilde{A^\bu}_Q\)};
    \node at (2,0.5) {\footnotesize\(\widetilde{A^\bu}_Q\)};
  \end{tikzpicture}
  \text{, or} 
    \begin{tikzpicture}[scale=0.59,baseline={([yshift=-1ex]current bounding box.center)}]
    \draw [thick, dashed] (-2,0) -- (0,-2);
    \draw [thick] (2,0)--(0,-2);
    \draw [thick] (0,-2) -- (0,-3);
    \draw [thick,dashed] (1,-1)--(0,0);
    \fill[white!50] (0.5,-1.5) circle (0.26cm);
    \node at (0.59,-1.5) {\tiny{$h'$}};
     \fill[gray!50] (1,-1) circle (0.26cm);
    \draw [thick] (1,-1) circle (0.26cm);
    \node at (1.02,-1) {\tiny$\rho$};
    \fill[gray!50] (0,-2) circle (0.26cm);
    \draw [thick] (0,-2) circle (0.26cm);
    \node at (0.02,-2) {\tiny$\rho$};
    \node at (0,-3.5) {\tiny\(\id_{\C[x,\theta]} \ot p\)};
    \node at (0,0.43) {\footnotesize\(\fp_Q\)};
    \node at (-2,0.43) {\footnotesize\(\fp_Q\)};
    \node at (2,0.5) {\footnotesize\(\widetilde{A^\bu}_Q\)};
  \end{tikzpicture}
  \text{, or} 
    \begin{tikzpicture}[scale=0.59,baseline={([yshift=-1ex]current bounding box.center)}]
    \draw [thick] (-2,0) -- (0,-2);
    \draw [thick] (2,0)--(0,-2);
    \draw [thick] (0,-2) -- (0,-3);
    \draw [thick,dashed] (1,-1)--(0,0);
    \fill[white!50] (0.5,-1.5) circle (0.26cm);
    \node at (0.59,-1.5) {\tiny{$h'$}};
     \fill[gray!50] (1,-1) circle (0.26cm);
    \draw [thick] (1,-1) circle (0.26cm);
    \node at (1.02,-1) {\tiny$\rho$};
    \fill[gray!50] (0,-2) circle (0.26cm);
    \draw [thick] (0,-2) circle (0.26cm);
    \node at (0.02,-2) {\tiny$m$};
    \node at (0,-3.5) {\tiny\(\id_{\C[x,\theta]} \ot p\)};
    \node at (0,0.43) {\footnotesize\(\fp_Q\)};
    \node at (-2,0.43) {\footnotesize\(\widetilde{A^\bu}_Q\)};
    \node at (2,0.5) {\footnotesize\(\widetilde{A^\bu}_Q\)};
  \end{tikzpicture},
\end{equation}
where $m=n_{0,2}$ is the multiplication on $\widetilde{A^\bu}$, and $\rho$ is the action of $\fp_Q$. But again, ${\sf w}_2(\rho)={\sf w}_2(m)=0$ and $\Im(h')\subset \Im(\id_{\C[x,\theta]}\ot h)\subset (v+\sum_{i\geq 2} w_i)\C[x,\theta,\lambda,v,w_2,\ldots]$. Thus, all trees of this form are ultimately killed by $\id_{\C[x,\theta]}\ot p$.

We then obtain that the transferred open-closed homotopy algebra structure on $\fp_Q\oplus A^\bu(\Oa_Y)_Q$ coincides with the one obtained through the pure spinor construction. Moreover, by the general theory of homotopy transfer \cite{Lodval}, the maps of \eqref{eq:AOYtwistretract} will lift to quasi-isomorphisms of open-closed homotopy algebras.
\end{proof}
\subsubsection{Spans of untwisted open-closed homotopy algebras}
\label{sec:spanuntwist}
It should be noted that it was nowhere assumed that $Q\neq0$. Hence, \cref{prop:AOYopenclosed} has the following immediate corollary.
\begin{cor}
  \label{cor:AOYopenclosed}
  $\fp\oplus \widetilde{A^\bu}$ and  $\fp\oplus A^\bu(\Oa_Y)$ are quasi-isomorphic as open-closed homotopy algebras.
\end{cor}
Now, \cref{thm:roofuntwisted} tells us that $\widetilde{A^\bu}$ and $C^\bu(\fn)$ are quasi-isomorphic as $A_\infty$-algebras. Let us sketch how this works, and see that the generalisation to open-closed homotopy algebras follows immediately.

Recall that, as a dg vector space
\begin{equation}
\widetilde{A^\bu}\cong\big( \C[x,\theta,\lambda,v,w_2,\ldots],\,\lambda\pdv{}{\theta}+v\pdv{}{x}-\lambda\theta\pdv{}{x}+d_{\tilde{\ft}}\big).  
\end{equation}
It was observed in \cite{cederwall2023canonical} that, since $\widetilde{A^\bu}$ is semifree it constitutes the Chevalley--Eilenberg cochains of a (nonminimal) $L_\infty$-algebra, $(\widetilde{A^\bu})^!$. As a cochain complex 
\begin{equation}
    \label{eq:koszuldualatilde}
    (\widetilde{A^\bu})^!\cong (\ft \overset{\mu_1}{\to}{\ft}[-1]\overset{0}{\to}\fn),
  \end{equation}
  where $\mu_1\sim d_\alpha[-1]\pdv{}{d_\alpha}+e_\mu[-1]\pdv{}{e_\mu}$. 
  It is thus clear that $\H^\bu((\widetilde{A}^\bu)^!,\mu_1)\cong \fn$ and it was moreover proved that $\fn$, with its $L_\infty$-algebra structure given by restriction, is in fact a minimal model of $\widetilde{A^\bu}$. 

 To see this we simply observe that, since  $\H^\bu((\widetilde{A^\bu})^!,\mu_1)\cong \fn$, there exists (\cf \cite[for example]{Berglund:0909.3485,CHUANG2019130}\footnote{this result is sometimes referred to as the ``tensor trick''.}) a strong multiplicative retract 
 \begin{equation}
    \begin{tikzcd}
       \label{eq:tildeACnretract}
\arrow[loop left]{r}{H} (\widetilde{A^\bu},\lambda\pdv{}{\theta}+v\pdv{}{x}) \arrow[r, shift left=1ex, "{P}"] 
 &  \arrow[l, shift left=1ex, "{I}"]  (C^\bu(\fn),0)  
\end{tikzcd}\,.
\end{equation}
We can choose a decomposition of $\widetilde{A^\bu}\cong C^\bu(\fn) \oplus (x+\theta+\lambda+v)\C[x,\theta,\lambda,v,w_2,\ldots]$, where elements of the last summand always contain at least one power of $x$, $\theta$, $\lambda$, or $v$. $P$ is then the projection onto the first summand, and $I$ the injection. Furthermore, we have, modulo coefficients,  $H=\theta\!\!\overset{(IP,\id)}{\pdv{}{\lambda}}+x\!\!\overset{(IP,\id)}{\pdv{}{v}}$, where the diacritic $(IP,\id)$ invokes the $(IP,\id)$-derivation property (\cf \cref{sec:hpl}). We can now add the rest of the differential on $\widetilde{A^\bu}$, that is, we consider a perturbation of the differential $\lambda\pdv{}{\theta}+v\pdv{}{x}$, of the form
\begin{equation}
  \label{eq:perturbation}
  d\coloneqq-\lambda\Gamma\theta\pdv{}{x}+d_{\tilde{\ft}}.
\end{equation}
By the homological perturbation lemma this induces a differential on $C^\bu(\fn)$, which, by the results of \cite{cederwall2023canonical}, yields the Chevalley--Eilenberg differential $d_\fn$ defining the $L_\infty$-algebra structure on $\fn$. Moreover, by the homological perturbation lemma, the resulting retract
 \begin{equation}
    \begin{tikzcd}
       \label{eq:tildeACnretract4}
\arrow[loop left]{r}{H'} (\widetilde{A^\bu},\tilde{d}) \arrow[r, shift left=1ex, "{P'}"] 
 &  \arrow[l, shift left=1ex, "{I'}"]  (C^\bu(\fn),d_\fn)  
\end{tikzcd}\,
\end{equation}
is a strong multiplicative retract and the deformed maps $I'$ and $P'$ constitute $L_\infty$-quasi-isomorphisms of the underlying $L_\infty$-algebras of $C^\bu(\fn)$ and $\widetilde{A^\bu}$ \cite{Berglund:0909.3485,CHUANG2019130}.

Now, since $\fp\oplus\widetilde{A^\bu}$ is an open-closed homotopy algebra, we can, analogous to the proof of \cref{prop:AOYopenclosed}, extend the deformation retract \eqref{eq:tildeACnretract4} by taking the direct sum on each side with $\fp$. We can then transfer open-closed homotopy algebras along this extended retract to obtain an open-closed homotopy algebra on $\fp\oplus C^\bu(\fn)$. In particular, since the retract \eqref{eq:tildeACnretract4} is strong and multiplicative, the induced $A_\infty$-algebra structure on $C^\bu(\fn)$ coincides with the binary one already defined.
As open-closed homotopy algebras are equivalent to $A_\infty$-algebras over $L_\infty$-algebras (see \cref{sec:openclosed}), the following theorem is then immediate.
\begin{theorem}
There is a span
\label{thm:rooftopenclosed}
    \begin{equation}
        \begin{tikzcd}[row sep = 1 ex]
          \label{eq:roofopenclosed}
          & \widetilde{A^\bu} \ar[rd] \ar[ld] &  \\
                A^\bu(\Oa_Y) & & C^\bu(\fn)
        \end{tikzcd},
      \end{equation}
of quasi-isomorphisms of $A_\infty$-algebras over the (strict) $L_\infty$-algebra $\fp$. 
\end{theorem}
\subsubsection{\texorpdfstring{$\widetilde{A}^\bu_Q$}{\~A{\_}Q} and \texorpdfstring{$C^\bu(\fn)_Q$}{C(n)\_Q}}
We will here generalise the right leg of \eqref{eq:roofopenclosed} to also allow for arbitrary twists. This yields a quasi-isomorphism of $A_\infty$-algebras over $L_\infty$-algebras between $\widetilde{A^\bu}_Q$ and $C^\bu(\fn)_Q$, where $C^\bu(\fn)_Q$ is $C^\bu(\fn)$ with a particular deformed differential, soon to be defined.

In the previous section we saw that, by considering the perturbation \eqref{eq:perturbation} of the differential on the left side of \eqref{eq:tildeACnretract}, we have an equivalence between $\widetilde{A^\bu}$ and $C^\bu(\fn)$. Now, instead of just considering the perturbation \eqref{eq:perturbation}, what if we also include the term coming from twisting? That is, what if we add the perturbation
\begin{equation}
  \label{eq:twistedperturbation}
  d_Q\coloneqq-\lambda\Gamma\theta\pdv{}{x}+d_{\tilde{\ft}}+ Q\pdv{}{\theta}+Q\Gamma\theta\pdv{}{x},
\end{equation}
to the left side of \eqref{eq:tildeACnretract} and see what it induces on $C^\bu(\fn)$?
\begin{prop}
  The differential induced on $C^\bu(\fn)$ by adding the perturbation \eqref{eq:twistedperturbation} to the left hand side of the retract \eqref{eq:tildeACnretract} is of the form
  \begin{equation}
    \label{eq:twistedfndif}
    d_{\fn}^Q\coloneqq d_\fn + \sum_{n\geq1} P\circ Q^{\alpha_1}\cdots Q^{\alpha_{n}}\pdv{^n}{\lambda^{\alpha_1}\cdots\partial\lambda^{\alpha_{n}}}\circ d_{\ft_1\circlearrowright\fn} \circ I\,,
  \end{equation}
  where $d_{\ft_1\circlearrowright\fn}$ are precisely those terms of $d_{\tilde{\ft}}$ encoding the brackets of $\tilde{\ft}$ which carry at least one element of $\ft_1$ and the rest in $\fn$.
\end{prop}
\begin{proof}
The new differential $d_{\fn}^Q$ is given as a sum of terms 
\begin{align}
d_{\fn}^Q=\sum_{n=0}^\infty d'_n,\qquad d'_n= P (d_Q H)^n d_Q I\;.
\end{align}
The term for $n=0$ gives back the Chevalley-Eilenberg differential $d_\fn$ on $C^\bullet(\fn)$. Indeed,  $d_Q= -\lambda\Gamma\theta \pdv{}{x}+d_\fn+(d_{\tilde{\ft}}-d_\fn)+Q\pdv{}{\theta}+Q\Gamma\theta\pdv{}{x}$. All terms but $d_\fn+Q\pdv{}{\theta}$ are zero under $P$, that is their image lies inside $(x+\theta+v+\lambda)\C[x,\theta,\lambda,v,w_2,\ldots]$, and $Q\pdv{}{\theta}$ composed with $I$ is zero.

Let us now consider the higher corrections $d'_n$, for $n\geq1$. The first correction looks like
\begin{equation}
    d'_1=P\circ(d_QH)\circ d_Q\circ I.
  \end{equation}
  As the image of $I$ will never contain any powers of $x,\theta,v,\lambda$, we can effectively replace the first occurence of $d_Q$ with $d_{\tilde{\ft}}$. Now, crucially, to have something nonzero under $P$ implies there can be no $x,\theta,v,\lambda$ in the expression. By working backwards one can try to identify which terms produce a possible nonzero image under $P$. The only part of $d_Q$ which does not produce any such terms is $Q\pdv{}{\theta}+d_\fn$. Hence
\begin{equation}
\label{eq: d1'}
    d_1'=P\circ(d_\fn+Q\pdv{}{\theta})(\theta\!\!\overset{(IP,\id)}{\pdv{}{\lambda}}+x\!\!\overset{(IP,\id)}{\pdv{}{v}})\circ d_{\tilde{\ft}} \circ I=P\circ(Q\pdv{}{\lambda})\circ d_{\tilde{\ft}}\circ I,
\end{equation}
where we used the fact that the composition $d_\fn H$ always contains $x$ or $\theta$, which is zero under $P$. Similarly, the composition $Q\pdv{}{\theta}\circ x\pdv{}{v}$ always carries an $x$ which is also zero under $P$. Moreover, the $(PI,\id)$-derivation property vanishes in the last step because of the identity $PI=\id_{C^\bu(\fn)}$.

The following lemma computes the general form for the higher terms of $d_\fn^Q$.
\begin{lemma}
\label{lem: higher corrections}
Let $b_n=P\circ (d_Q H)^n$, for $n\geq 1$. Then we have,
\begin{equation}
b_n=P\circ Q^{\alpha_1}\cdots Q^{\alpha_n}\pdv{^n}{\lambda^{\alpha_1}\cdots\partial\lambda^{\alpha_n}}. 
\end{equation}
\end{lemma}
\begin{proof}
  The proof is done by induction. The base case $n=1$ was shown in \eqref{eq: d1'}. Assume now $b_n=P\circ Q^{\alpha_1}\cdots Q^{\alpha_n}\pdv{^n}{\lambda^{\alpha_1}\cdots\partial\lambda^{\alpha_n}}$.
  Then we have, 
\begin{equation}\begin{split}
    b_{n+1}=&P\circ(d_Q H)^{n+1}=p\circ(d_Q H)^n(d_QH)=b_n(d_QH)=\\
    =&P\circ Q^{\alpha_1}\cdots Q^{\alpha_n}\pdv{^n}{\lambda^{\alpha_1}\cdots\partial\lambda^{\alpha_n}}\Big(\!-\!\lambda\Gamma\theta\pdv{}{x}+d_{\tilde{\ft}}+Q\pdv{}{\theta}+Q\Gamma\theta\pdv{}{x}\Big)\circ\Big(\theta\!\!\overset{(IP,\id)}{\pdv{}{\lambda}}+x\!\!\overset{(IP,\id)}{\pdv{}{v}}\Big)\\
    =&P\circ Q^{\alpha_1}\cdots Q^{\alpha_{n+1}}\pdv{^n}{\lambda^{\alpha_1}\cdots\partial\lambda^{\alpha_{n+1}}},
\end{split}
\end{equation}
where the $(PI,\id)$-derivation property vanishes in the last step because of the identity $PI=\id_{C^\bu(\fn)}$.
\end{proof}
By \cref{lem: higher corrections}, we thus have that the total transferred differential takes the form:
\begin{equation}
  \label{eq:deformeddiff}
  d_\fn^Q= d_\fn + \sum_{n\geq1} P\circ Q^{\alpha_1}\cdots Q^{\alpha_{n}}\pdv{^n}{\lambda^{\alpha_1}\cdots\partial\lambda^{\alpha_{n}}}\circ d_{\tilde{\ft}} \circ I.
  \end{equation}
  However, the expression after $d_{\tilde{\ft}}\circ I$ must carry at least one $\lambda$, in order to create something nontrivial, and it cannot contain any $v^\mu$ as these will be killed under $P$. Thus, the only surviving terms are those encoding brackets which take at least one element of $\ft_1$ and the rest in $\fn$. Hence, we retrieve \eqref{eq:twistedfndif}.
\end{proof}

\subsubsection{ Koszul duals of twisted canonical multiplets}
Recall that $\tilde{\ft}\cong \ft\oplus \fn$, as graded vector spaces. Moreover, if $Q$ is a Maurer--Cartan element of $\ft$ then $Q$ is also a Maurer--Cartan element of $\tilde{\ft}$. The twist of $\tilde{\ft}$ by $Q$ is then the $L_\infty$-algebra $(\tilde{\ft}_Q,\mu_k^{\tilde{\ft}_Q})$ defined on the vector space $\tilde{\ft}_Q=\tilde{\ft}$, with brackets
   \begin{equation}
    \label{eq:mutildeftQ}
    \mu_k^{\tilde{\ft}_Q}(x_1,\ldots,x_k)= \sum_{n\geq0}\frac1{n!}\mu_{k+n}^{\tilde{\ft}}(Q,\ldots,Q,x_1,\ldots,x_k).
  \end{equation}
  As $\fn$ forms an $L_\infty$-algebra ideal inside $\tilde{\ft}$ it is clear from \eqref{eq:mutildeftQ} that the underlying graded vector space $\fn\subset\tilde{\ft}_Q$, with its $L_\infty$-algebra structure defined by restriction of \eqref{eq:mutildeftQ}, also forms an ideal inside $\tilde{\ft}_Q$. Let us denote this $L_\infty$-algebra $\fn_Q$. The cochains $C^\bu(\fn_Q)$ of $\fn_Q$ are then the symmetric algebra $\Sym^\bu(\fn^*[-1])$, equipped with a differential $d_{\fn_Q}$. Let us choose a basis $n_a$ for $\fn_Q=\fn$ and let $N^a$ denote the shifted duals of $n_a$, then the differential on $C^\bu(\fn_Q)$ can be written as
  \begin{equation}
    \label{eq:dn_Q}
    d_{\fn_Q}=\sum_{\substack{k\geq 1\\n\geq 0}}Q^{\alpha_1}\cdots Q^{\alpha_n}F^b_{\alpha_1\cdots\alpha_n a_1\cdots a_k}N^{a_1}\cdots N^{a_k}\pdv{}{N^b},
  \end{equation}
  where the $F$'s are, modulo signs and combinatorial factors, the structure constants of $\tilde{\ft}$, so that $F^b_{\alpha_1\cdots\alpha_n a_1\cdots a_k}n_b\sim[d_{\alpha_1},\ldots,d_{\alpha_n},n_{a_1},\ldots,n_{a_n}]$.
  
Let us now see what brackets the extra term of \eqref{eq:twistedfndif} defines. Using the same bases $n_a$ and $N^a$, we have
\begin{equation}
  \label{eq:explicitdceterms}
  d_{\ft_1\circlearrowright\fn} =\sum_{i,k\geq 1}F_{\alpha_1\cdots\alpha_i a_1\cdots a_k}^b\lambda^{\alpha_1}\cdots\lambda^{\alpha_i}N^{a_1}\cdots N^{a_k}\pdv{}{N^b}.
\end{equation}
Substituting \eqref{eq:explicitdceterms} in \eqref{eq:twistedfndif} forces $i=n$ and we get a sum over $n$ and $k$. The result is 
\begin{equation}
    \label{eq:explicitdeformed}
    d_\fn^Q=d_\fn+ \sum_{n,k\geq1}Q^{\alpha_1}\cdots Q^{\alpha_n}F_{\alpha_1\cdots\alpha_n a_1\cdots a_k}^bN^{a_1}\cdots N^{a_k}\pdv{}{N^b},
  \end{equation}
  from which we can read off that the ``new'' $L_\infty$-algebra structure on $\fn$, defined by $d_{\fn}^Q$, is precisely that of $\fn_Q$.
We thus obtain the following corollary.
  \begin{cor}
  The $L_\infty$-algebra Koszul dual to $(C^\bu(\fn),d_{\fn}^Q)$ is the $L_\infty$-algebra $\fn_Q$; the subalgebra ideal of $\tilde{\ft}_Q$ generated by $\fn$.
\end{cor}

Now, by the same reasoning as in \cref{sec:spanuntwist}, the homological perturbation lemma tells us there exists a strong multiplicative retract
 \begin{equation}
    \begin{tikzcd}
       \label{eq:tildeACnretract2}
\arrow[loop left]{r}{\tilde{H}} (\widetilde{A^\bu}_Q,\tilde{d}_Q) \arrow[r, shift left=1ex, "{\tilde{P}}"] 
 &  \arrow[l, shift left=1ex, "{\tilde{I}}"]  (C^\bu(\fn_Q),d_{\fn_Q})  
\end{tikzcd}\,,
\end{equation}
along which (after taking the direct sum with the retract identifying $\fp_Q$ with itself) we can transfer open-closed homotopy algebra structures. Since the retract is in fact strong and multiplicative, that is, $\tilde{I}$ and $\tilde{P}$ are (strict) algebra morphisms, we have in particular that the induced $A_\infty$-algebra structure on $C^\bu(\fn)$ is the binary one already defined. We have thus established the following theorem.
\begin{theorem}
There is a span
\label{thm:maintwisted}
    \begin{equation}
        \begin{tikzcd}[row sep = 1 ex]
          \label{eq:rooftwistopenclosed}
          & \widetilde{A^\bu}_Q \ar[rd] \ar[ld] &  \\
                A^\bu(\Oa_Y)_Q & & C^\bu(\fn_Q)
        \end{tikzcd},
      \end{equation}
of quasi-isomorphisms of  $A_\infty$-algebras over the (strict) $L_\infty$-algebra $\fp_Q$. 
\end{theorem}

\begin{remark}
 Let us briefly also consider $\widetilde{A^\bu}_Q$; this is again a semifree cdga, which means it constitutes the Chevalley--Eilenberg cochains of an $L_\infty$-algebra. The underlying vector space of this $L_\infty$-algebra is again ${\ft}[-1]\oplus \ft\oplus \fn$. From the equation of $\tilde{d}_Q$ we can read off the $L_\infty$-algebra structure;
\begin{equation}
    \tilde{d}_Q= v^\mu\pdv{}{x^\mu}+\lambda^\alpha\pdv{}{\theta^\alpha}-\lambda^\alpha\Gamma_{\alpha\beta}^\mu\theta^\beta\pdv{}{x^\mu}+Q^\alpha\pdv{}{\theta^\alpha}+Q^\alpha\Gamma_{\alpha\beta}^\mu\theta^\beta\pdv{}{x^\mu}+d_{\tilde{\ft}}.
  \end{equation}
  In particular, we see that we have a term $Q^\alpha\pdv{}{\theta^\alpha}$ which does not preserve the augmentation (\cf \cref{rem:curved}), and so
  $(\widetilde{A^\bu}_Q)^!$
  will be a curved $L_\infty$-algebra. In particular, the curvature is
  \begin{equation}
      \mu_0(1)=Q^\alpha d_\alpha[-1].\\
  \end{equation}
\end{remark}

\section*{Acknowledgements}
This work was supported by the Leverhulme Research Project Grant RPG--2021--092. The author would like to thank Charles Young, Hyungrok Kim, Ingmar Saberi, Martin Cederwall, Jakob Palmkvist, and Tommaso Franzini for valuable discussions. Special thanks to Charles Young and Tommaso Franzini for their helpful comments on the draft.

\end{document}